\documentclass[11pt]{article}

\usepackage[bookmarks]{hyperref}
\usepackage{amssymb}
\usepackage{amsmath}
\usepackage{amsthm}
\usepackage{stmaryrd}
\usepackage[affil-it]{authblk}
\usepackage{graphicx,color,colordvi}
\usepackage{bbm}
\usepackage{cleveref}

\usepackage{pgfplots}

\def\ra{\rangle}
\def\la{\langle}

\def\openone{\leavevmode\hbox{\small1\kern-3.8pt\normalsize1}}

\def\II{\mathbb I}

\def\CC{\mathbb{C}}
\def\RR{\mathbb{R}}
\def\ZZ{\mathbb{Z}}

\def\NN{\mathbb{N}}

\def\PP{\mathbb{P}}

\def\h{\mathfrak h}

\newtheorem{theorem}{Theorem}
\newtheorem{lemma}{Lemma}
\newtheorem{proposition}{Proposition}
\newtheorem{corollary}{Corollary}

\theoremstyle{definition}
\newtheorem{definition}{Definition}

\def\ketbra#1#2{|#1\rangle\langle #2|}

\newcommand{\bP}{{\overline{P}}}

\newcommand{\bfi}{{\textbf{i}}}
\newcommand{\bfj}{{\textbf{j}}}
\newcommand{\bfk}{{\textbf{k}}}
\newcommand{\bfx}{{\textbf{x}}}

\def\reff#1{(\ref{#1})}
\def\eps{\varepsilon}

\newcommand{\supp}{\mathop{\rm supp}\nolimits}

\newcommand{\tr}{\mathop{\rm Tr}\nolimits}

\newcommand{\spec}{{\rm sp}}

\newcommand{\hrho}{\hat{\rho}}
\newcommand{\hsigma}{\hat{\sigma}}
\newcommand{\hq}{\hat{q}}
\newcommand{\hr}{\hat{r}}

\newcommand{\cA}{{\cal A}}
\newcommand{\cB}{{\cal B}}
\newcommand{\cD}{{\cal D}}

\newcommand{\cF}{{\cal F}}

\newcommand{\cH}{{\cal H}}

\newcommand{\cP}{\mathcal{P}}

\newcommand{\cX}{{\cal X}}
\newcommand{\cO}{{\cal O}}

\newcommand{\es}{{e_{\rm{sym}}}}
\newcommand{\est}{{e^*_{\rm{sym}}}}
\newcommand{\rw}{{\rm{w}}}
\newcommand{\tnd}{\triangle_n D}
\def\d{\mathrm{d}}
\def\e{\mathrm{e}}

\usepackage{graphicx}
\usepackage{setspace}
\usepackage{verbatim}
\usepackage{subfig}

\theoremstyle{definition}

\theoremstyle{remark}
\newtheorem{remark}{Remark}
\newtheorem{condition}{Condition}
\numberwithin{equation}{section}

\DeclareRobustCommand\openone{\leavevmode\hbox{\small1\normalsize\kern-.33em1}}

\newcommand{\id}{\rm{id}}
\newcommand{\be}{\begin{equation}}
\newcommand{\ee}{\end{equation}}
\newcommand{\bea}{\begin{eqnarray}}
\newcommand{\eea}{\end{eqnarray}}
\newcommand{\beas}{\begin{eqnarray*}}
	\newcommand{\eeas}{\end{eqnarray*}}

\begin{document}
	\bibliographystyle{abbrv}
	
	\title{Second-order asymptotics for quantum hypothesis testing in settings beyond i.i.d. - quantum lattice systems and more}
	\author[a]{Nilanjana Datta}
	\author[b]{Yan Pautrat}
	\author[a]{Cambyse Rouz\'e}
	\affil[a]{\small Statistical Laboratory, Centre for Mathematical Sciences, University of Cambridge, Cambridge~CB30WB, UK}
	\affil[b]{\small Laboratoire de Math\'ematiques d'Orsay,
		Univ. Paris-Sud, CNRS, Universit\'e
		Paris-Saclay,  91405~Orsay, France}
	\maketitle

	
	\begin{abstract}
		
		Quantum Stein's Lemma is a cornerstone of quantum statistics and concerns the problem of correctly identifying a quantum state, given the knowledge that it is one of two specific states ($\rho$ or $\sigma$). It was originally derived in the asymptotic i.i.d.~setting, in which arbitrarily many (say, $n$) identical copies of the state ($\rho^{\otimes n}$ or $\sigma^{\otimes n}$) are considered to be available. In this setting, the lemma states that, for any given upper bound on the probability $\alpha_n$ of erroneously inferring the state to be $\sigma$, the probability $\beta_n$ of erroneously inferring the state to be $\rho$ decays exponentially in $n$, with the rate of decay converging to the relative entropy of the two states. The second order asymptotics for quantum hypothesis testing, which establishes the speed of convergence of this rate of decay to its limiting value, was derived in the i.i.d.~setting independently by Tomamichel and Hayashi, and Li. We extend this result to settings beyond i.i.d.. Examples of these include Gibbs states of quantum spin systems (with finite-range, translation-invariant interactions) at high temperatures, and quasi-free states of fermionic lattice gases.

	\end{abstract}
	\section{Introduction}\label{sec_intro}
	
	\subsection*{Quantum Hypothesis Testing}
	Quantum hypothesis testing concerns the problem of discriminating between two different quantum states\footnote{It is often referred to as {\em{binary}} quantum hypothesis testing, to distinguish it from the case in which there are more than two states.}. It is fundamentally different from its classical counterpart, in which one discriminates between 
	two different probability distributions, the difference arising essentially from the non-commutativity of quantum states of physical systems, which are at the heart of quantum information. Hence, discriminating between different states of a quantum-mechanical system is of paramount importance in quantum information-processing tasks. 
	In the language of hypothesis testing, one considers two hypotheses -- the {\em{null hypothesis}} $H_0 : \rho$ and the {\em{alternative hypothesis}} $H_1 :\sigma$, where $\rho$ and $\sigma$ are two quantum states. In an operational setting, say Bob receives a state $\omega$ with the knowledge that either $\omega=\rho$ or $\omega=\sigma$. His goal is then to infer which hypothesis is true, i.e., which state he has been given, by means of 
	a measurement on the state he receives. The measurement is given most generally by a POVM $\{T, \mathbb{I} - T\}$ where $0\le T\le \mathbb{I}$. Adopting the nomenclature from classical hypothesis testing, we refer to $T$ as a {\em{test}}. 
	The probability that Bob correctly guesses the state to be~$\rho$ is then 
	equal to $\tr(T\rho)$, whereas his probability of correctly guessing
	the state to be $\sigma$ is $\tr((\mathbb{I}-T)\sigma)$. Bob can erroneously infer the state to be $\sigma$ when it is actually $\rho$ or vice versa. The corresponding error probabilities are referred to as the {\em{type I error}} and {\em{type II error}} respectively. They are given as follows:
	\begin{align} \alpha(T) &:= \tr\left(( \mathbb{I} -  T)\rho\right), \quad \beta(T) := \tr\left(T \sigma\right), 
	\end{align}
	where $\alpha(T)$ is the probability of accepting $H_1$ when 
	$H_0$ is true, while $\beta(T)$ is the probability
	of accepting $H_0$ when $H_1$ is true. Obviously, there is a trade-off between the two error probabilities, and
	there are various ways to jointly optimize them, depending on whether or not 
	the two types of errors are treated on an equal footing. In the setting of {\em{symmetric hypothesis testing}}, one minimizes the total probability of error $\alpha(T)+\beta(T)$, whereas in {\em{asymmetric hypothesis testing}} one minimizes the type II error under a suitable constraint on the type I error.
	
	Quantum hypothesis testing was originally studied in the {\em{asymptotic i.i.d.~setting}} in which Bob is provided not with just a single copy of the state but with multiple (say $n$) 
	identical copies of the state, and he is allowed to do a joint measurement on all these copies. The optimal asymptotic performance in the different settings is quantified by the following exponential decay rates, evaluated in the asymptotic limit ($n \to \infty$):
	$(i)$ the optimal exponential decay rate of the sum of type I and type II errors, $(ii)$ the optimal exponential decay rate of the type II error under the 
	assumption that the type I error decays with a given exponential
	speed, and $(iii)$ the optimal exponential decay rate of the type II error 
	under the assumption that the type I error remains bounded.
	The first of these corresponds to the symmetric setting while the other two to the asymmetric setting; $(i)$ is given by the quantum Chernoff bound \cite{Aetal07, NussbaumSzkola}, $(ii)$ is given by the Hoeffding bound \cite{HM7, N06, OH04} while $(iii)$ is given by quantum Stein's lemma \cite{HP91, ON00}. In this paper we will restrict attention to the asymmetric setting of quantum Stein's Lemma and its refinement, and hence we elaborate on it below.
	
	\subsection*{Quantum Stein's lemma and its refinement: asymptotic i.i.d.~setting}
	Suppose that the state $\omega_n$ that Bob receives is either the state $\rho_n:=\rho^{\otimes n} $ or the state
	$\sigma_n:=\sigma^{\otimes n}$, with $\rho$ and $\sigma$ being states on a finite-dimensional Hilbert space $\cH$. In the setting of quantum Stein's lemma, the quantity of interest is the \emph{optimal asymptotic type II error exponent}, which is given, for any $\eps\in(0,1)$, by 
	\begin{align*}
	\lim_{n \to \infty}- \frac{1}{n} \log \beta_n(\eps),\qquad {\hbox{where}} \,\, \beta_n(\eps):=\inf_{0 \leq T_n \leq \mathbb{I}_n}\{\beta(T_n)|\alpha(T_n)\le \eps\}.
	\end{align*}
	Here $\alpha(T_n)= \tr[(\mathbb{I}_n - T_n)\rho_n]$ and $\beta(T_n) = \tr[T_n \sigma_n]$, with $\mathbb{I}_n$ being the
	identity operator acting on $\cH^{\otimes n}$. Quantum Stein's lemma (see \cite{HP91,ON00}) establishes that 
	$$\lim_{n \to \infty}- \frac{1}{n} \log \beta_n(\eps)=D(\rho||\sigma) \quad \forall \, \eps \in (0,1),$$ 
	where $D(\rho||\sigma)$ denotes the quantum relative entropy defined in \Cref{qrel}. Moreover, it states that the minimal asymptotic type I error jumps discontinuously from $0$ to $1$ as the asymptotic type II error exponent crosses the value $D(\rho||\sigma)$ from below: see \Cref{fig1}, which is a plot of the function 
	\begin{equation} \label{eq_defalpha1}
	\alpha_\infty^{(1)}(t_1)=\inf\big\{ \limsup_n \alpha(T_n)\, |\, -\liminf_n\frac1{n}\log \beta(T_n) \geq t_1\big\}.
	\end{equation}
	However, this discontinuous dependence of the minimal asymptotic type~I error on the asymptotic type II error exponent is a manifestation of the coarse-grained analysis underlying the Quantum Stein's lemma, in which only the linear term (in~$n$) of the {\em{type~II error exponent}} $(-\log \beta(T_n))$ is considered.
	
	\begin{figure}[h]
		\centering
		\includegraphics[width=0.5\textwidth]{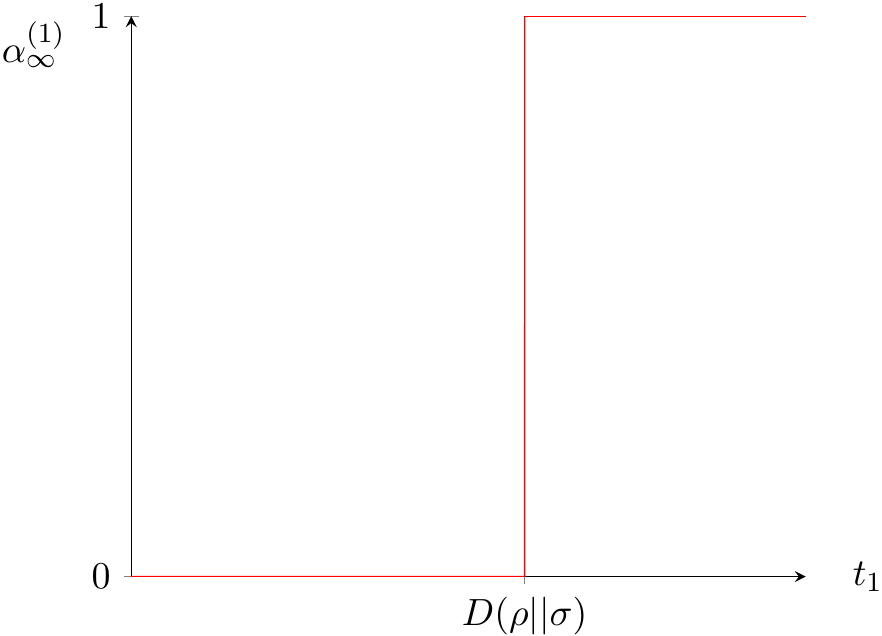}
		\caption{The minimal asymptotic type I error $\alpha_\infty^{(1)}(t_1)$, defined in \Cref{eq_defalpha1}.}\label{fig1}
	\end{figure}
	More recently, Li \cite{L14} and Tomamichel and Hayashi \cite{TH13} independently showed that this discontinuity vanishes under a more refined analysis of the 
	type II error exponent, in which its second order (i.e.~order $\sqrt{n}$) term 
	is retained, in addition to the linear term. This analysis is referred to as the {\em{second order asymptotics}} for (asymmetric) quantum hypothesis testing, since it involves the evaluation of $\left(-\log\beta_n(\eps)\right)$ up to second order. It was proved in \cite{TH13, L14} to be given by:
	\begin{align}\label{soa}
	-\log\beta_n(\eps)=n\,D(\rho||\sigma)+\sqrt{n\, V(\rho||\sigma)}\,\Phi^{-1}(\eps)+\mathcal{O}(\log n),
	\end{align}
	where $\Phi$ denotes the cumulative distribution function (c.d.f.)~of a standard normal distribution, and $V(\rho||\sigma)$ is called the {\em{quantum information variance}} and is defined in \Cref{info-var}. The above expansion implies that, if the minimal type~II error exponent is constrained to have $n\,D(\rho||\sigma)$ as its first order term and $\sqrt{n}\,t_2$ as its second order term, then the minimal type~I error is given by $\Phi(t_2/\sqrt{V(\rho||\sigma)})$.  Hence 
	the minimal asymptotic type I error varies smoothly between $0$ and $1$ as $t_2$ increases from $-\infty$ to $+\infty$: see \Cref{fig2}, which is a plot of the function 
	\begin{equation} \label{eq_defalpha2}
	\alpha_\infty^{(2)}(t_2)=\inf\big\{ \limsup_n \alpha(T_n)\, |\, -\liminf_n\frac1{\sqrt{n}}\big(\log \beta_n(T_n) + n \, D(\rho||\sigma)\big)\geq t_2\big\}.
	\end{equation}
	\begin{figure}[h]
		\centering
		\includegraphics[width=0.5\textwidth]{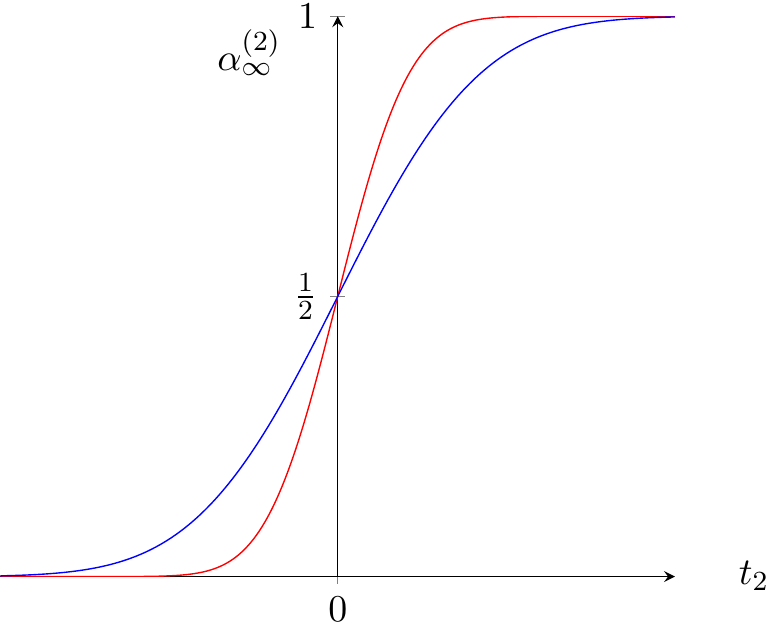}
		\caption{The minimal asymptotic type I error $\alpha_\infty^{(2)}(t_2)$, defined in \Cref{eq_defalpha2}, for two different values of $V(\rho||\sigma)$ (red: $V(\rho||\sigma)=1$, blue: $V(\rho||\sigma)=2$).}\label{fig2}
	\end{figure}
	The Gaussian c.d.f. $\Phi$ arises from the Central Limit Theorem (CLT), or rather from its refinement, the Berry-Esseen Theorem (see e.g.~\cite{Feller2}), which gives the rate of convergence of the distribution of the scaled sum of i.i.d.~random variables to a normal distribution.
	
	The study of second order asymptotics is a key step towards understanding a fundamental problem of both theoretical and practical interest, which is to determine how quickly the behaviour of a finite system approaches that of its asymptotic limit. Second order asymptotics have been obtained for various classical- (\cite{Strassen62,HTan14,WH}, ; see also the review \cite{T14} and references therein) and quantum information theoretic tasks (initiated in \cite{TH13, L14}; see \cite{LD14} and references therein for later works, as well as \cite{WRG14}). However, all results obtained in the quantum  case pertain to the i.i.d.~setting in which the underlying resource is assumed to be uncorrelated. In this paper, we obtain second order asymptotics for quantum settings beyond i.i.d.~which had thus far remained a challenging, open problem\footnote{It has been done for mixed source coding in \cite{LD14} but there too the problem can essentially be reduced to the i.i.d.~setting.}. Moreover, we do it for the task of quantum hypothesis testing which underlies various other information-theoretic tasks, such as the transmission of classical information through a quantum channel.

	\subsection*{Our contribution}
	In this paper, we extend the study of second order asymptotics of the {\em{type II error exponent}} $\left(-\log\beta_n(\eps)\right)$ of asymmmetric quantum hypothesis testing to settings beyond i.i.d.. In this case, the state $\omega_n$  which Bob receives is either 
	$\rho_n$ or $\sigma_n$, but the latter are not necessarily of the tensor power (i.e.~i.i.d.) form. More precisely, we consider two sequences $\hrho:=(\rho_n)_{n\in\NN}$ and $\hsigma:=(\sigma_n)_{n\in\NN}$, where for every $n$, $\rho_n$ and $\sigma_n$ are states (on a finite-dimensional Hilbert space $\cH_n$) which need not be of the form $\rho^{\otimes n}$ and $\sigma^{\otimes n}$, respectively. We still assume, however, that for increasing $n$ the amount of information that can be extracted by performing a test increases linearly with a specified sequence of positive weights $(w_n)_{n\in {\mathbb{N}}}$, satisfying $w_n \to \infty$ as $n \to \infty$. In the case of $n$ i.i.d.~copies, i.e. if $\rho_n=\rho^{\otimes n}$ and $\sigma_n=\sigma^{\otimes n}$, we choose $w_n=n$. Interesting examples of quantum states which fall within our setting are Gibbs states of quantum spin systems (with finite-range, translation-invariant interactions) at high temperature and quasi-free states of a fermionic lattice gas. The example of i.i.d.~states is included in our setting and so we can recover the result of \cite{TH13} and \cite{L14}. 
	
	The first order asymptotics in this general setting was studied by Hiai, Mosonyi and Ogawa \cite{Hiai2008}, where the authors introduced the analysis of the hypothesis testing problem using the existence and differentiability of the asymptotic logarithmic moment generating function, which leads to a proof of quantum Stein's lemma  for various classes of correlated states (see \cite{Hiai2009}, \cite{Jaksicetal}). It is also worth mentioning that various extensions of quantum Stein's lemma for non ~i.i.d. states were previously treated in \cite{Bjelakovic2008,Bjelakovic2004,Hiai2007,Hiai2008,Hiai1994,MOSONYI,Mosonyi2008a}. Special cases of the second order asymptotics that we derive were obtained in \cite{HT14} and \cite{TT13}:
	in Lemma 26 of \cite{TT13} the case in which $\hrho$
	is a sequence of product states and $\hsigma$ is a sequence of i.i.d 
	states was considered, whereas in \cite{HT14} $\hrho$ was 
	chosen to be i.i.d.~while $\hsigma$ was taken to be  
	some sequence of so-called {\em{universal states}} (see Lemma 1 of \cite{HT14}). 
	
	Our main results are given by \Cref{mainresult} and \Cref{coro_ideal} of \Cref{sec_main}. 
	The essence of our strongest result (which is given by \Cref{coro_ideal}), can be easily conveyed
	through the example of quasi-free states of fermions on a one-dimensional lattice $\ZZ$. Suppose
	the sequence of states $\rho_n$ and $\sigma_n$ arising in our hypothesis testing problem are Gibbs states of these fermions, corresponding to the same Hamiltonian but at two different temperatures, restricted to a finite subset $\Lambda_n$ of the lattice. Even though these are
	non-i.i.d.~states, we prove that the minimal type I error has the same {\em{functional dependence}} on the second order term
	of the type II error exponent as that for i.i.d.~case shown in \Cref{fig2}. Note, however, that the relative entropy $D(\rho||\sigma)$ is replaced by the relative entropy rate $d(\hrho, \hsigma)$, the quantum information variance, $V(\rho||\sigma)$, is replaced by the quantum information variance rate, $v(\hrho,\hsigma)$ (defined in \Cref{drho} and \Cref{dvee} respectively), and $n$ is replaced by the lattice size $|\Lambda_n|$.
	
	\section{Notations and Definitions}\label{preliminaries}
	In this section let $\cH$ denote a finite-dimensional Hilbert space, $\cB(\cH)$ denote the algebra of linear operators acting on $\cH$ and $\cB_{sa}(\cH)$ denote the set of self-adjoint operators. Let $\cP(\cH)$ be the set of positive semi-definite operators on $\cH$ and $\cP_{+}(\cH) \subset  \cP(\cH)$ the set of (strictly) positive operators. Further, let $\cD(\cH):=\lbrace\rho\in\cP(\cH)\mid \tr\rho=1\rbrace$ denote the set of density matrices on $\cH$; we will use the terms ``density matrix" and ``state" interchangeably. We denote the support of an operator $A$ as ${\mathrm{supp}}(A)$ and the range of a projection operator $P$ as ${\mathrm{ran}}(P)$. Let $\mathbb{I}\in\cP(\cH)$ denote the identity operator on $\cH$, and $\id:\cB(\cH)\mapsto \cB(\cH)$ the identity map on operators on~$\cH$.
	
	Any element $A$ of $\cB_{sa}(\cH)$ has a \textit{spectral decomposition} of the form
	$A = \sum_{\lambda \in \spec(A)} \lambda \,P_\lambda(A),$ where $\spec(A)$ denotes the spectrum of $A$, and $P_\lambda(A)$ is the projection operator corresponding to the eigenvalue $\lambda$. We denote by $A_+$ the positive part of $A$.  
	 More precisely,
	\[ A_+ := \sum_{\lambda \in \spec(A)\cap \RR_+} \lambda\, P_\lambda(A)\]
	
	For any $\rho, \sigma \in \cP(\cH)$, and for any $s \in [0,1]$, 
	we define the following quantity which plays a central role in our proofs:
	\begin{equation}\label{phi-s}
	\Psi_s(\rho|\sigma):= \log\tr\left(\rho^s\sigma^{1-s}\right).
	\end{equation}
	If $\rho$ and $\sigma$ have orthogonal supports, then $\Psi_s(\rho|\sigma)= - \infty$. Note that for $s \in [0,1)$,
	\begin{equation}
	\Psi_s(\rho|\sigma)= (s-1) D_s(\rho||\sigma),
	\end{equation}
	where $D_s(\rho||\sigma)$ is the quantum relative R\'enyi entropy of order $s$. It is known that
	\begin{equation}
	\lim_{s \to 1} D_s(\rho||\sigma) 
	= D(\rho||\sigma),
	\end{equation}
	where $D(\rho||\sigma)$ is the quantum relative entropy,
	\begin{equation} \label{qrel}
	D(\rho||\sigma):=
	\tr \rho (\log \rho - \log \sigma),
	\end{equation}
	if $\supp(\rho) \subseteq \supp(\sigma)$, and equal to $+ \infty$ otherwise.
	Henceforth, we assume for simplicity that all states are faithful, so that $\supp(\rho) =\supp(\sigma)$.
	We also define the quantum information variance
	\be\label{info-var}
	V(\rho||\sigma):= \tr \big(\rho(\log \rho - \log \sigma)^2\big) - D(\rho||\sigma)^2.
	\ee
	The following identities can be easily verified.
	:
	\begin{lemma}\label{lem1}
		For any $\rho, \sigma \in \cD(\cH)$ with $\supp(\rho)=\supp(\sigma)$ we have
		\begin{equation*}
		\left.\frac{\d}{\d s}\Psi_s(\rho|\sigma)\right|_{s=0} = - D(\sigma||\rho) \qquad 
		\left.\frac{\d}{\d s} \Psi_s(\rho|\sigma)\right|_{s=1} = D(\rho || \sigma) \qquad
		\left.\frac{\d^2}{\d s^2} \Psi_s(\rho|\sigma)\right|_{s=1} = V(\rho || \sigma).
		\end{equation*}
		\noindent Moreover, for $s \in [0,1]$, $\Psi_s^{\prime\prime}(\rho|\sigma) \geq 0$ and hence
		$\Psi_s(\rho|\sigma)$ is convex in $s$.
	\end{lemma}
	
	As in \cite{Jaksicetal}, we use {\em{relative modular operators}} in our proofs and intermediate results. Relative modular operators were introduced originally by Araki, and this allowed him to extend the notion of relative entropy to arbitrary states on a C*-algebra (see \cite{Araki76, Araki77, OP}). We also refer to Petz's papers \cite{Petz1985} and \cite{Petz1986} for a discussion on the relation between the relative modular operator and R\'{e}nyi divergences. We do not work with infinite-dimensional systems in the present paper; however, the language of modular theory will prove convenient even in the finite-dimensional case. In addition, extensions to infinite dimensions will be natural once statements and proofs are expressed in the language of modular theory (this will be done in a future paper). In particular the crucial relation \eqref{phi-s-inner} (given below) will hold without change in the infinite-dimensional setting.
	
	To define relative modular operators on a finite-dimensional operator algebra $\cB(\cH)$, we start by equipping $\cA=\cB(\cH)$ with a Hilbert space structure through the Hilbert-Schmidt scalar product: for $A,B$ two elements of $\cA$ we let $\langle A,B\rangle := \tr A^* B$. We define a map $\pi:\cB(\cH)\to \cB(\cA)$ by $\pi(A):X\mapsto AX$, i.e.~$\pi(A)$ is the map acting on $\cA=\cB(\cH)$ by left multiplication by $A$. This map is linear, one-to-one and has in addition the properties $\pi(AB)=\pi(A)\pi(B)$, $\pi(A^*)=\pi(A)^*$, where here $\pi(A)^*$ denotes the adjoint of the map $\pi(A)$ defined through the 
	relation $\langle X,\pi(A)(Y)\rangle = \langle \pi(A)^*X,Y\rangle$, and the following identity between operator norms $\|\pi(A)\|_{\cB(\cA)}=\|A\|_{\cB(\cH)}$.
	We will therefore identify $A$ with $\pi(A)$ and, because of the identity $\pi(A)X=AX$, will simply write~$A$ for $\pi(A)$ (even though $\pi(A)$ is a linear map on $\cA$, and $A$ is not!).
	
	For any $\rho\in \cD(\cH)$, we denote $\Omega_\rho:=\rho^{1/2}\in \cB_{sa}(\cH)$. We then have the identity
	\begin{equation}
	\tr(\rho A) = \langle \Omega_\rho, A \Omega_\rho \rangle \qquad \mbox{for all }A\in\cA\label{eq_GNS}
	\end{equation}
	even though, again, the right-hand side should be written $\langle \Omega_\rho, \pi(A) \Omega_\rho \rangle$. This is nothing but a simple case of the well-known GNS representation (see e.g. Section 2.3.3 of \cite{BR1}).
	
	From now on, for simplicity of exposition, we only consider faithful states, i.e.~any state~$\rho$ will be such that $\supp(\rho)=\cH$ and in particular any two states will satisfy $\supp(\rho)=\supp(\sigma)$. We then define the \textit{relative modular operator} $\Delta_{\rho|\sigma}$ to be the map
	\begin{equation}\label{eq_defDelta}
	\begin{array}{cccc}
	\Delta_{\rho|\sigma}: & \cA & \to & \hspace{-1em} \cA \\
	& A & \mapsto & \rho A \sigma^{-1}
	\end{array}
	\end{equation}
	Let us denote by $J$ the antilinear operator on $\cA$ defined by $J:A\mapsto A^*$. This map is anti self-adjoint and anti-unitary:
	\begin{align}\label{antiJ}
	\forall A,B\in\mathcal{B}(\mathcal{H}),\quad \langle A,J(B)\rangle=\overline{\langle J(A),B\rangle},\quad \langle J(A),J(B)\rangle=\overline{\langle A,B\rangle},
	\end{align}
	 and we easily obtain\footnote{We can also derive the relation $J\Delta^{1/2}_{\rho|\sigma} \, A\Omega_\sigma = A^* \Omega_\rho$ for any $A\in \cA$, which is precisely the relation defining~$J$ and $\Delta_{\rho|\sigma}$ in the general algebraic case.}  the relation $J \Delta_{\rho|\sigma} J = \Delta_{\sigma|\rho}^{-1}$.
	Note that \eqref{eq_defDelta} defines $\Delta_{\rho|\sigma}$ for any $\rho$, $\sigma$ in $\cP_{+}(\cH)$ and not only for states.
	
	As a linear operator on $\mathcal B(\cH)$, $\Delta_{\rho|\sigma}$ is positive and its spectrum $\spec(\Delta_{\rho|\sigma})$ consists of the ratios of eigenvalues $\lambda/\mu$, $ \lambda \in \spec(\rho)$, 
	$ \mu \in \spec(\sigma)$. For any $x\in \spec (\Delta_{\rho|\sigma})$, the corresponding spectral projection is the map
	\begin{equation} \label{eq_specprojDelta}
	\begin{array}{cccc}
	P_{x}(\Delta_{\rho|\sigma}): & \cA & \to & \hspace{-1em} \cA \\
	& A & \mapsto & \underset{\lambda\in \spec(\rho),\mu\in\spec(\sigma):\lambda/\mu=x}{\sum}P_\lambda(\rho) A  P_\mu(\sigma)
	\end{array}
	\end{equation}
	Further, for any $s \in [0,1]$,
	$\Delta_{\rho|\sigma}^s(A) = \rho^s A \sigma^{-s}$. For any $\rho, \sigma \in \cP_{+}(\cH)$ and $s \in [0,1]$, the quantity $\Psi_s(\rho|\sigma)$ can be expressed in terms of the 
	relative modular operator $\Delta_{\rho|\sigma}$ as follows
	\be\label{phi-s-inner}
	{\Psi_s(\rho|\sigma)} =  \log \langle\Omega_\sigma, \Delta_{\rho|\sigma}^s\Omega_\sigma \rangle, \qquad \text{where }\Omega_{\sigma}:=\sigma^{1/2}.
	\ee
	In addition, let $\mu_{\rho|\sigma}$ denote the spectral measure for $- \log \Delta_{\rho|\sigma}$ with respect to $\Omega_\sigma:= \sigma^{1/2}$, i.e.~the probability measure such that for any bounded measurable function $f$, 
	\begin{align} \label{spec-meas}
	\langle\Omega_\sigma, f(-\log \Delta_{\rho|\sigma})\Omega_\sigma\rangle = \int\! f(x) \,\d\mu_{\rho|\sigma}(x) \equiv {\mathbb{E}}[f(X)],
	\end{align}
	where $X$ is a random variable of law $\mu_{\rho|\sigma}$ (see e.g. Sections VII and VIII of \cite{RS1}).
	We then have in particular
	\begin{align} \label{eq_fundamental}
	\Psi_s(\rho|\sigma) = \log \left(\int e^{-sx}\,\d\mu_{\rho|\sigma}(x)\right) \equiv \log {\mathbb{E}}[e^{-sX}]
	.
	\end{align}
	The relation \eqref{eq_fundamental} 
	plays a key role in our proof since it links the quantity $\Psi_s$ to the cumulant generating function of a classical random variable, and therefore allows us to employ the tools of 
	classical probability theory. 
	The quantity $\Psi_s$ can also be expressed in terms of the well-known Nussbaum-Szko\l a distributions (see~\cite{NussbaumSzkola}). For two density matrices $\rho, \sigma\in\cD(\cH)$ with spectral decompositions
	$$\rho= \sum_\lambda \lambda P_\lambda(\rho), \quad \sigma = \sum_\mu \mu P_\mu(\sigma),$$
	these distributions are given by $(p_{\lambda,\mu})_{\lambda, \mu}$ and $(q_{\lambda,\mu})_{\lambda, \mu}$, where
	$$ p_{\lambda,\mu}= \lambda  \tr\big(P_{\lambda}(\rho)\, P_{\mu}(\sigma)\big), \quad q_{\lambda,\mu}= \mu  \tr\big(P_{\lambda}(\rho)\, P_{\mu}(\sigma)\big).$$
	There is of course a connection between the Nussbaum-Szko\l a distributions and relative modular operators. Assume for simplicity that all ratios $\lambda/\mu$ are distinct and consider a random variable $Z$ which takes values $\lambda/\mu$ with probability $q_{\lambda,\mu}$. Then using \eqref{eq_specprojDelta} one can easily verify that
	\begin{align} \PP(Z=\lambda/\mu) = q_{\lambda,\mu} &= \mu \, \tr\big(P_{\lambda}(\rho)\, P_{\mu}(\sigma)\big)\nonumber\\
	&= \la \Omega_\sigma, P_{\lambda/\mu}(\Delta_{\rho|\sigma})\, \Omega_\sigma \ra \nonumber\\
	&= \la \Omega_\sigma, {\mathbf{1}}_{\{\lambda/\mu\}}(\Delta_{\rho|\sigma})\, \Omega_\sigma \ra, \label{laweq}
	\end{align}
	where ${\mathbf{1}}_{\{\lambda/\mu\}}$ denotes the indicator function on the singleton $\{\lambda/\mu\}$, {\em{i.e.}}~${\mathbf{1}}_{\{\lambda/\mu\}}(x)$ is equal to $1$ 
	when $x = \lambda/\mu$ and equal to $0$ else. This follows from the fact that, since $\Delta_{\rho|\sigma}$ is self-adjoint, the spectral theorem implies that  $\mathbf{1}_{\{\lambda/\mu\}}(\Delta_{\rho|\sigma})=\sum_{x\in\spec(\Delta_{\rho|\sigma})}\mathbf{1}_{\{\lambda/\mu\}}(x)P_x(\Delta_{\rho|\sigma})$. \Cref{laweq} implies that for any bounded measurable function $f$, 
	$${\mathbb{E}}\left[f(Z)\right] = \la \Omega_\sigma, f(\Delta_{\rho|\sigma})\, \Omega_\sigma \ra$$
	and hence the law of $Z$ is the law of $\Delta_{\rho|\sigma}$ with respect to $\Omega_\sigma:= \sigma^{1/2}$. This in turn implies that for any bounded measurable function $f$, 
	$${\mathbb{E}}\left[f(- \log Z)\right] = \la \Omega_\sigma, f(- \log \Delta_{\rho|\sigma})\, \Omega_\sigma \ra.$$
	Hence the law of $-\log Z$ is precisely $\mu_{\rho|\sigma}$. As mentioned earlier, the advantage of the construction via modular operators is that it extends directly to the infinite-dimensional case.

	\section{Second order asymptotics for hypothesis testing} \label{sec_main}
	
	In this section, we state and prove our main results. Our general framework will be that of two sequences of states, $\hrho=(\rho_n)_{n\in\NN}$ and $\hsigma=(\sigma_n)_{n\in\NN}$, such that for every $n$, $\rho_n$ and $\sigma_n$ are elements of $\cD(\cH_n)$, where $\cH_n$ is a finite-dimensional Hilbert space. We always assume that~$\rho_n$ and $\sigma_n$ are faithful, that is,
	$\supp(\rho_n) = \supp(\sigma_n) = \cH_n$. We also fix a sequence of \textit{weights}~$\hat w=(w_n)_{n\in {\mathbb{N}}}$ which we assume is an increasing sequence of positive numbers satisfying $\lim_{n \to \infty} w_n = \infty$.
	
	The quantity whose second order asymptotic expansion we wish to evaluate is the {\em{type II error exponent}}, which is given for any $\eps\in(0,1)$ by $- \log \beta_n(\eps)$, where
	\[
	\beta_n(\eps):=\inf_{0 \leq T_n \leq \mathbb{I}_n}\{\beta(T_n)|\alpha(T_n)\le \eps\}.
	\]
	Here $\alpha(T_n):= \tr\left((\mathbb{I}_n - T_n)\rho_n\right)$ and $\beta(T_n) := \tr\left(T_n \sigma_n\right)$, are the {\em{type I}} and {\em{type II errors}}, in testing $\rho_n$ \textit{vs.} $\sigma_n$ and $\mathbb{I}_n$ is the identity operator acting on $\cH_n$.
	
	\begin{remark}
		We restrict our results to the case of faithful states $\rho_n$, $\sigma_n$ only to make our exposition more transparent. Simple limiting arguments show that all our results remain valid in the case in which $\supp(\rho_n) \subseteq \supp(\sigma_n)$.
	\end{remark}
	
	Our main results hold under the following conditions on the sequences of 
	states $(\rho_n,\sigma_n)_{n\in {\mathbb{N}}}$ with respect to the weights $(w_n)_{n\in {\mathbb{N}}}$:
	\begin{condition}\label{cond-Bryc}
		There exists $r>0$ such that  
		\begin{enumerate}
			\item for every $n$, the function $E_n(z)=\Psi_{1-z}(\rho_n|\sigma_n)$ originally defined for $z\in[0,1]$ can be extended to an analytic function in the complex open ball $B_\CC(0,r)$,
			\item for every $x$ in the real open ball $B_{\mathbb R}(0,r)$, the following limit exists: \[E(x)=\lim_{n\to\infty} \frac1{w_n}\, E_n (x),\]
			\item one has the uniform bound
			\[\sup_{n\in\mathbb N}\sup_{z\in B_{\CC}(0,r) } \frac 1{w_n} |E_n (z)|<+\infty. \]
		\end{enumerate}
	\end{condition}
	Vitali's Theorem, stated in \Cref{app_Vitali}, implies that, if \Cref{cond-Bryc} holds, then $E$ can be extended to an analytic function on $B_{\CC}(0,r)$, and the derivatives $\frac 1{w_n} E_n'(0)$ and $\frac 1{w_n} E_n''(0)$ converge to $E'(0)$ and $E''(0)$ respectively. By \Cref{lem1}, this shows that the following definition is meaningful under \Cref{cond-Bryc}:
	\begin{definition}\label{def-hat}
		Let $\hat{\rho} = (\rho_n)_{n \in \NN}$ and $\hat{\sigma} = (\sigma_n)_{n \in \NN}$ denote two sequences of states that satisfy \Cref{cond-Bryc} with respect to a sequence of weights $(w_n)_{n\in {\mathbb{N}}}$. Their \emph{quantum relative entropy rate} is defined as
		\begin{align}\label{drho}
		d(\hrho, \hsigma) := \lim_{n \to \infty} \frac{1}{w_n} D(\rho_n || \sigma_n),
		\end{align}
		and their \emph{quantum information variance rate} is defined as
		\begin{align}\label{dvee}
		v(\hrho, \hsigma) := \lim_{n \to \infty} \frac{1}{w_n} V(\rho_n || \sigma_n).
		\end{align}
	\end{definition}
	
	Our main result on second order asymptotics is given by the following theorem, which we prove in \Cref{sec_proofmainthm}.
	\begin{theorem}\label{mainresult}
		Fix $\eps\in(0,1) $, and let $\hat{\rho} = (\rho_n)_{n \in \mathbb{N}}$ and $\hat{\sigma} = (\sigma_n)_{n \in \mathbb{N}}$ denote two sequences of states that satisfy \Cref{cond-Bryc} with respect to a sequence of weights $(w_n)_{n\in\NN}$. Define
		\[ t_2^*(\eps)=\sqrt{v(\hrho, \hsigma)\,}\Phi^{-1}(\eps)\]
		where $\Phi$ is the cumulative distribution function of a standard normal distribution and $v(\hrho, \hsigma)$ is the quantum information variance rate defined through \Cref{dvee}.
		Then for any $t_2>t_2^*(\eps)$ there exists a function $f_{t_2}(x)\underset{+\infty}{=}o(\sqrt x)$ such that for any $n\in\NN$:
		\begin{equation} \label{eq_maineq1}
		-\log\beta_n(\eps) \leq  D(\rho_n||\sigma_n) + \sqrt{w_n}\,  t_2 + f_{t_2}(w_n),
		\end{equation}
		and for any $t_2<t_2^*(\eps)$, any function $f(x)\underset{+\infty}{=}o(\sqrt x)$, for $n$ large enough:
		\begin{equation} \label{eq_maineq2}
		-\log\beta_n(\eps) \geq  D(\rho_n||\sigma_n) + \sqrt{w_n}\,  t_2+ f(w_n).
		\end{equation}
	\end{theorem}
	Let us comment on the above theorem, and explain why it captures the second order behaviour of $- \log \beta_n(\eps)$. For this discussion, let us denote
	\begin{align}
	\tnd = D(\rho_n||\sigma_n)- w_n \,d(\hrho, \hsigma).
	\end{align}
	By \Cref{def-hat} we know that $D(\rho_n||\sigma_n) =w_n \,d(\hrho, \hsigma) + o(w_n)$, so that $-\log\beta_n(\eps)=w_n \,d(\hrho, \hsigma) + o(w_n)$. The second order term of $-\log\beta_n(\eps)$ will therefore depend on the relative magnitude of $\tnd$ and $\sqrt{w_n}$. The simplest situation is summarized in the following corollary.
	\begin{corollary} \label{coro_ideal}
		Fix $\eps\in(0,1) $, and let $\hat{\rho} = (\rho_n)_{n \in \mathbb{N}}$ and $\hat{\sigma} = (\sigma_n)_{n \in \mathbb{N}}$ denote two sequences of states that satisfy \Cref{cond-Bryc} with respect to a sequence of weights $(w_n)_{n\in\NN}$. In addition, assume that
		\[D(\rho_n||\sigma_n) = w_n \, d(\hrho,\hsigma) + o(\sqrt{w_n}).\]
		Let
		\[t_2^*(\eps):=\sqrt{v(\hrho, \hsigma)\,}\Phi^{-1}(\eps).\]
		Then for any $t_2>t_2^*(\eps)$ there exists a function $f_{t_2}(x)\underset{+\infty}{=}o(\sqrt x)$ such that for any $n\in\NN$:
		\[-\log\beta_n(\eps) \leq  w_n \, d(\hrho,\hsigma) + \sqrt{w_n}\,  t_2 + f_{t_2}(w_n),\]
		and for any $t_2<t_2^*(\eps)$, any function $f(x)\underset{+\infty}{=}o(\sqrt x)$, for $n$ large enough:
		\[-\log\beta_n(\eps) \geq w_n \, d(\hrho,\hsigma) + \sqrt{w_n}\,  t_2 + f(w_n).\]\end{corollary}
	
	The proof of \Cref{coro_ideal} is immediate from \Cref{mainresult}. It says that if $(\rho_n,\sigma_n)_{n\in {\mathbb{N}}}$ satisfies \Cref{cond-Bryc} with respect to a sequence of weights $(w_n)_{n\in {\mathbb{N}}}$ and $\tnd=o(\sqrt{w_n})$, then for any $\delta>0$, we have 
	\begin{align*}
	w_n \, d(\hrho,\hsigma)  + \sqrt{w_n}\,  \big(t^*_2(\eps)-\delta\big) + o(\sqrt{w_n})\leq -\log\beta_n(\eps)\leq  w_n \, d(\hrho,\hsigma)  + \sqrt{w_n}\,  \big(t^*_2(\eps)+\delta\big) + o(\sqrt{w_n}).
	\end{align*}
	If on the other extreme $\sqrt{w_n}=o(\tnd)$ then it is
	\[-\log\beta_n(\eps)= w_n \,d(\hrho, \hsigma) + \tnd + o(\tnd)\]
	that captures the second order of $-\log\beta_n(\eps)$.
	
	\begin{remark}
		An equivalent statement to the above results is the following: define for $t_2\in\RR$
		\begin{align*}
		\tilde \alpha_\infty^{(2)}(t_2)&=\inf\big\{ \limsup_n \alpha(T_n)\, |\, -\liminf_n\frac1{\sqrt{w_n}}\big(\log \beta(T_n) + D(\rho_n||\sigma_n)\big)\geq  t_2\,;\, 0 \le T_n \le {\mathbb{I}}_n\big\},\\
		\alpha_\infty^{(2)}(t_2)&=\inf\big\{ \limsup_n \alpha(T_n)\, |\, -\liminf_n\frac1{\sqrt{w_n}}\big(\log \beta(T_n) + w_n d(\hrho,\hsigma)\big)\geq  t_2\,;\, 0 \le T_n \le {\mathbb{I}}_n\big\}.
		\end{align*}
		Then, under the assumptions of \Cref{mainresult} (respectively  \Cref{coro_ideal}), one has
		\[
		\tilde\alpha_\infty^{(2)}(t_2)=\Phi\big(t_2/\sqrt{v(\hrho,\hsigma)}\big)\quad \mbox{(respectively}\ \alpha_\infty^{(2)}(t_2)=\Phi\big(t_2/\sqrt{v(\hrho,\hsigma)}\big)\mathrm{)}.\]
	\end{remark}
	\subsection{Ingredients of the proof of \Cref{mainresult}}
	A theorem proved by Bryc \cite{bryc} (stated as \Cref{theo_Bryc} below), the lower bound in \Cref{boundsD} and \Cref{prop_1} are the ingredients required to prove \Cref{mainresult}. Bryc's theorem can be viewed as a generalization of the Central Limit Theorem
	for sequences of random variables which are not necessarily i.i.d. It was proven originally in \cite{bryc}. \Cref{boundsD} involves bounds on the minimum error probability of symmetric hypothesis testing and was proven in \cite{Jaksicetal}. 
	
	The following version of Bryc's theorem is adapted from \cite{JOPP}. A full proof can be found in \Cref{app_Vitali}.
	\begin{theorem}[Bryc's theorem]\label{theo_Bryc}
		Let $(X_n)_{n\in\NN}$ be a sequence of random variables, and $(w_n)_{n\in\NN}$ a sequence of weights, i.e.~positive numbers satisfying $\lim_{n \to \infty} w_n= \infty$. If there exists $r>0$ such that the following conditions hold:
		\begin{itemize}
			\item for any $n \in \NN$, the function $H_n(z):=\log\mathbb{E}\left[\operatorname{e}^{-zX_n}\right]$ is analytic  in the complex open ball~$B_\CC(0,r)$,
			\item the limit $H (x) = \lim_{n\to\infty} \frac 1 {w_n}\, H _n(x)$ exists for any $x \in  (-r,r)$,
			\item we have the uniform bound $\sup_{n\in\mathbb{N}} \sup_{z\in B_\CC(0,r)} \frac1{w_n}\,|H_n(z)|<+\infty$,
		\end{itemize}
		then $H$ is analytic on $B_\CC(0,r)$ and we have the convergence in distribution\footnote{We use the notation $\mathcal{N}(\mu,\sigma^2)$ to denote a normal distribution of mean $\mu$ and variance~$\sigma^2$.} as $n\to\infty$
		\begin{align*}
		\frac{X_n+ H_n'(0)}{\sqrt{w_n\,}}\overset{\mathrm{d}}\longrightarrow\mathcal{N}\big(0,{H''(0)}\big),
		\end{align*}
	\end{theorem}
	\begin{remark}
		If $X_n$ is of the form $X_n=A_1+\ldots+ A_n$ for a sequence $(A_n)_{n\in\NN}$ of independent, identically distributed random variables, then for any $n$ we have $\frac1n H_n(z)=H(z)$. Therefore, the conditions of the above  theorem are satisfied (with respect to the weights $w_n=n$) as soon as $H(z)$ can be extended to an analytic function in a complex neighbourhood of the origin. Note that this is a stronger condition than the existence of a second moment.
	\end{remark}
	
Let us comment on the use we make of \Cref{theo_Bryc}.
	Let $\hat{\rho} = (\rho_n)_{n \in \mathbb{N}}$ and $\hat{\sigma} = (\sigma_n)_{n \in \mathbb{N}}$ denote two sequences of states satisfying \Cref{cond-Bryc}. Setting $s=1-z$, with $z\in [0,1]$ in \Cref{eq_fundamental} and using the cyclicity of the trace function we obtain for $\Omega_{\rho_n}:=\rho_n^{1/2}$
\begin{align}
	\Psi_{1-z}(\rho_n|\sigma_n)=\log\langle \Omega_{\rho_n},\Delta^z_{\sigma_n|\rho_n} \Omega_{\rho_n} \rangle=\log \mathbb{E}\left[ \e^{-zX_n}\right],\label{psi1-z}
	\end{align}
	where $(X_n)_{n \in \mathbb{N}}$ denotes a sequence of random variables with law $\mu_{\sigma_n|\rho_n}$, defined analogously to \Cref{spec-meas}. Then we have
	\begin{equation}
	H_n(z):=\log \mathbb{E}[e^{-zX_n}] = \Psi_{1-z}(\rho_n|\sigma_n) \label{eq_relationXPsi}
	\end{equation}
	and \Cref{cond-Bryc} implies that $(X_n)_{n\in\NN}$ satisfies the assumptions of Bryc's theorem. Since
	\begin{gather*}
	H_n'(0)= \left.-\frac{\d}{\d z} \Psi_{z} (\rho_n|\sigma_n)\right|_{z=1}=-D(\rho_n||\sigma_n),\\
	H_n''(0)=\left.\frac{\d^2}{\d z^2}\Psi_{z}(\rho_n|\sigma_n)\right|_{z=1}= V(\rho_n||\sigma_n),
	\end{gather*}
	\Cref{cond-Bryc} and \Cref{theo_Bryc} imply the convergence in distribution 
	\begin{align}\label{convg}  
	Y_n:= \frac{X_n-D(\rho_n||\sigma_n)}{\sqrt{w_n}}& \overset{\mathrm{d}}\longrightarrow\mathcal{N}\big(0,v(\hat{\rho},\hat{\sigma})\big).
	\end{align}
	The above convergence \reff{convg} plays a fundamental role in our proof.
	\medskip
	
	Another key ingredient of our proof is \Cref{boundsD}, which gives an upper bound and a lower bound on the minimum total probability of error in symmetric hypothesis testing. If the states corresponding to the two hypotheses are $\rho$ and $\sigma$, then the latter is defined as:
	\be\label{e-symm}
	\est(\rho, \sigma) := \inf_{0 \le T \le \mathbb{I}} \es(\rho, \sigma, T),
	\ee
	where
	\[
	\es(\rho, \sigma, T) := \alpha(T) +\beta(T),
	\]
	denotes the total probability of error under a test $T$. 
	\medskip
	
	\noindent
	\begin{remark}In the Bayesian setting, one can assign prior probabilities to the two states, i.e.~assume that the state is $\rho$ with probability $p$, and $\sigma$ with probability $(1-p)$ for some $p \in (0,1)$. The aim is then to minimize the quantity $\left[p\, \tr((\mathbb{I}-T)\rho)+ (1-p)\tr(T \sigma)\right]$, which is $\es(p\rho, (1-p)\sigma, T)$. For convenience we absorb the scalar factors $p$ and $(1-p)$ into $\rho$ and $\sigma$, so that the latter are no longer of trace $1$. In order to accommodate this scenario, as well as for particular applications in the proof of our main result, in the following we consider
		the quantities $\est(A, B)$ and $\es(A, B)$ for arbitrary operators
		$A, B\in \cP_{+}(\cH)$. We also use the fact that $\est$ is symmetric in its arguments, {i.e.}~$\est(A, B)=\est(B,A)$.
	\end{remark}
	
	The infimum in \Cref{e-symm} is attained for a test given by the projector $(A-B)_+$, as is given by the quantum extension of the Neyman-Pearson lemma \cite{Holevo72, Helstrom76} which states that
	for any $A, B \in \cP_{+}(\cH)$ 
	\begin{align}\label{lemmaproj}
	\operatorname{\est}(A,B)=\operatorname{\es}\big(A,B,(A-B)_+\big).
	\end{align}	
	For any $\gamma \in {\mathbb{R}}$, the orthogonal projection $T^\gamma:= (A-e^\gamma B)_+$ onto the positive part of $(A-e^\gamma B)$, is called the quantum Neyman-Pearson test for $A$ \textit{vs.}~$B$, of rate $\gamma$. Using \Cref{lemmaproj}, one can show that the tests $T^{\gamma}$ are optimal in the sense that for any other test $T'$, $$\alpha(T')\le \alpha(T^{\gamma})\Rightarrow \beta(T')\ge \beta(T^{\gamma}).$$ 
	In the classical framework, i.e.~when $[\rho,\sigma]=0$ then for any $\gamma \in {\mathbb{R}}$, $T^\gamma$ 
	is $P_{\RR_+}\big(\log(\rho\sigma^{-1})-\gamma\mathbb I\big)$, and we recover the Neyman-Pearson tests considered in classical hypothesis testing.
	More precisely, writing $\rho$ and $\sigma$ in the basis in which they are both diagonal, so that their spectral decompositions read
	\[ \rho=\sum_i \lambda_i P_i\qquad \mu=\sum_i \mu_i P_i,\]
	we have 
	\[T^\gamma = \sum_{i|\log(\lambda_i/\mu_i)\geq \gamma}  P_i,\,\]
	and
	\begin{equation*} \label{eq_connecNPRV}
	\beta(T^{\gamma})
	=\sum_{i|\log(\lambda_i/\mu_i)\geq \gamma}\mu_{i}
	=\mathbb{P}(X\geq \gamma),
	\end{equation*}
	where $X$ is a random variable which takes values $\log({\lambda_i}/{\mu_i})$ with probabilities $\mu_i$. This is particularly useful in the i.i.d.~setting in which we have $n$ identical copies of $\rho$ and $\sigma$. In this case the random variable $X \equiv X_n$ can be expressed as a sum of $n$ i.i.d.~random variables, and we can use the full power of probability theory\footnote{$(-X_n)$ is known as the \textit{observed log-likelihood ratio} or \textit{information content} random variable and can be used to define the classical Neyman-Pearson tests.}. However, this equivalence between quantum and classical Neyman-Pearson tests is not valid in the non-commutative case, since $\big(\rho - e^\gamma \sigma\big)_+ \ne \big(\log(\rho\sigma^{-1})-\gamma {\mathbb{I}}\big)_+$.
	
	The following lemma from \cite{Jaksicetal} (whose proof we include in \Cref{app_boundseNP}), allows us to circumvent this problem; even when $[\rho, \sigma] \ne 0$, it provides bounds on $\est(\rho, \sigma)$ in terms of expectations of a classical random variable, which, 
	in the case of $n$ identical copies of the underlying state, can be expressed as a sum of $n$ i.i.d.~random variables as above. 
	We only use the lower bound on $\est(\rho, \sigma)$ in proving \Cref{mainresult}, but also state the upper bound for sake of completeness.
	
	\begin{lemma}\label{boundsD}
		For any $A, B \in \cP_{+}(\cH)$ and $s\in [0,1]$:
		\begin{align}\label{ineq}
		\langle \Omega_B,\Delta_{A|B}(1+\Delta_{A|B})^{-1}\Omega_B\rangle\le \operatorname{\est}(A,B)\le \langle\Omega_B,\Delta^s_{A|B}\Omega_B\rangle.
		\end{align}	
		where
		\begin{align*}
		\langle \Omega_B,\Delta_{A|B}(1+\Delta_{A|B})^{-1}\Omega_B\rangle	&=\sum_{(\lambda,\mu)\in \operatorname{sp}(A)\times\operatorname{sp}(B)}\frac{\lambda/\mu}{1+\lambda/\mu}\tr\big(B P_\lambda(A)P_\mu(B)\big),\\
		\langle\Omega_B,\Delta^s_{A|B}\Omega_B\rangle
		&= \tr A^s B^{1-s}.
		\end{align*}
		Equivalently, \begin{align}\label{ineq_1}
		\mathbb{E}\Big[ \frac{1}{1 + e^X}\Big] \le \operatorname{\est}(A,B)\le \mathbb{E}[e^{-sX}],
		\end{align}	
		where $X$ is a random variable with law $\mu_{A|B}$ (defined through \eqref{spec-meas}).
	\end{lemma}
	\begin{proof}
		See \Cref{app_boundseNP}.
	\end{proof}
	\smallskip

	The following result follows immediately from \Cref{boundsD}, and is also proven in \Cref{app_boundseNP}.
	\begin{corollary}\label{cor1}
		For any $\rho, \sigma \in {\cal{P}}_{+}(\cH)$ and any $\upsilon, \theta \in\mathbb{R}$,
		\begin{align}\label{markovineq}
		\operatorname{\est}(\sigma,\rho\operatorname{e}^{-\theta})\ge \frac{\e^{-\theta}}{1+\e^{\upsilon-\theta}}\, {\mathbb{P}}(X \le \upsilon),
		\end{align}
		where $X$ is a random variable with law $\mu_{\sigma|\rho}$ defined similarly to \Cref{spec-meas}.
	\end{corollary}
	\medskip
	
	The last ingredient in the proof of \Cref{mainresult} is the following explicit construction, which is due to Li (in \cite{L14}). We give a full proof in \Cref{app_Prop1}.
\begin{proposition}\label{prop_1}
	Let $\rho$, $\sigma$ be two states in ${\mathcal{D}}({\mathcal{H}})$. For any $L>0$ there exists a test $T$ such that
	\begin{equation} 
	\tr \,\rho(1-T) \leq  \big\langle{\Omega_\rho}, P_{[0,L)} (\Delta_{\rho|\sigma})\, \Omega_\rho\big\rangle \qquad
	\mbox{and}\qquad \tr\,\sigma T \leq  L^{-1}, \label{eq_alphabeta}
	\end{equation}
where $ P_{[0,L)} (\Delta_{\rho|\sigma}) := \sum_{x \in [0, L)} P_x (\Delta_{\rho|\sigma})$, with $ P_x (\Delta_{\rho|\sigma})$ being the projection operator defined through (\ref{eq_specprojDelta}).
\end{proposition}
	
	\subsection{Proof of \Cref{mainresult}} \label{sec_proofmainthm}
	
	The proof is divided into two parts: \Cref{eq_maineq1} and \Cref{eq_maineq2}.
	
	\paragraph{Proof of \Cref{eq_maineq1}} By continuity and monotonicity of $\Phi$, it suffices to prove that for any $\delta>0$, there exists $g_\delta(x)\underset{+\infty}{=}o(\sqrt x)$ such that for any $n\in\NN$:
	\begin{align}\label{part1proof} 
	-\log\beta_n(\eps)\le  D(\rho_n||\sigma_n)+\sqrt{w_nv(\hrho||\hsigma)}\,\Phi^{-1}(\eps+\delta)+o(\sqrt{w_n}).
	\end{align}
	For any test $T_n$, and real numbers $\theta_n$ and $v_n$ (to be specified later), \Cref{cor1} as well as the symmetry of $\est$ give 
	\begin{align*}
	\operatorname{e}^{-\theta_n}\tr \big((\mathbb{I}_n-T_n)\rho_n\big)+\tr (T_n\sigma_n)\geq  e^*_{sym}(\sigma_n,\rho_n\operatorname{e}^{-\theta_n})
	\ge \frac{ \operatorname{e}^{-\theta_n}}{1+\operatorname{e}^{v_n-\theta_n}} \mathbb{P}(X_n\le v_n),
	\end{align*}
	where $X_n$ is a random variable of law $\mu_{\sigma_n|\rho_n}$. Hence, for any sequence of tests $(T_n)_{n\in\NN}$ such that $\tr\big((\mathbb{I}_n-T_n)\rho_n\big)\le \eps$, we have
	\begin{align}
	\tr (T_n\sigma_n)&\ge \operatorname{e}^{-\theta_n}\Big(\frac{\mathbb{P}(X_n\le v_n)}{1+\operatorname{e}^{v_n-\theta_n}} -\tr \big((\mathbb{I}_n-T_n)\rho_n\big)\Big)\nonumber\\
	&\ge \operatorname{e}^{-\theta_n}\Big(\frac{\mathbb{P}(X_n\le v_n)}{1+\operatorname{e}^{v_n-\theta_n}} -\eps\Big). \label{berryesseenchanges}
	\end{align}
	Then by relation \eqref{convg} and \Cref{dvee}, we know that the random variable
	\begin{align}
	Y_n:=\frac{X_n- D(\rho_n||\sigma_n)}{\sqrt{w_n v(\hrho, \hsigma)}}
		\end{align}
	converges in law to a standard normal distribution.
	Let us choose
	$$v_n:= D(\rho_n||\sigma_n)+\sqrt{w_n\,v(\hrho, \hsigma)}\,\Phi^{-1}(\eps+\delta),$$
	with $\delta>0$, so that:
	\begin{align*}
\mathbb{P}(X_n\le v_n)=\mathbb{P}(Y_n\le \Phi^{-1}(\eps+\delta))\underset{n\to\infty}{\to}\Phi(\Phi^{-1}(\eps+\delta))=\eps+\delta>\eps.
	\end{align*}	
	Therefore, there exists $\eps'>\eps$ and $n_0\in\NN$ such that for all $n\ge n_0$, $\mathbb{P}(X_n\le v_n)\ge \eps'>\eps $. Take $\theta_n:=v_n+h(w_n)$, for some positive function $h$ such that $h(x)\underset{+\infty}{=}o(\sqrt{x})$ but $h(w_n)\to\infty$ as $n\to\infty$. Hence, again for $n$ large enough,
	\begin{align*}
	\operatorname{e}^{-\theta_n}\Big(\frac{\mathbb{P}(X_n\le v_n)}{1+\operatorname{e}^{v_n-\theta_n}} -\eps\Big) \ge \operatorname{e}^{-\theta_n}\, \big( \eps'(1+e^{-h(w_n)})^{-1}- \eps\big),
	\end{align*}
	and since the quantity in parenthesis on the right hand side of the above inequality is bounded below, and \eqref{berryesseenchanges} is true for any sequence of tests $(T_n)_{n \in {\mathbb{N}}}$ such that $\tr\big(\rho_n(\mathbb{I}_n-T_n)\big)~\le~\eps$, we have for $n$ large enough
	\begin{align*}
	-\log\beta_n(\eps) \le
	D(\rho_n||\sigma_n)+\sqrt{w_nv(\hrho, \hsigma)}\,\Phi^{-1}(\eps+\delta)+\tilde g_\delta(w_n),
	\end{align*}
	with $\tilde g_\delta \underset{+\infty}{=}o(\sqrt x)$, and modifying $\tilde g_\delta$ for a finite number of values gives the desired inequality.
	\medskip
	
	\paragraph{Proof of \Cref{eq_maineq2}}
	
	For each $n$ in $\NN$, let
	\[L_n = \exp \big(D(\rho_n||\sigma_n) + \sqrt{w_n\,} t_2+ f(w_n) \big) \]
	with $f(x)\underset{+\infty}{=}o(\sqrt{x})$, and consider the test $T_n$ associated by Proposition \ref{prop_1} with $L_n$. Then relations \eqref{eq_alphabeta} imply
	\begin{equation} \label{eq_betan}
		-\log \beta(T_n)  \geq D(\rho_n||\sigma_n) + \sqrt{w_n\,} t_2+ f(w_n),
	\end{equation}
	and
	\begin{align*}
		\alpha(T_n)
		&\leq  \big\langle{\Omega_{\rho_n}}, P_{(0,L_n)} (\Delta_{\rho_n|\sigma_n})\, \Omega_{\rho_n}\big\rangle,\\
		&\leq  \big\langle{\Omega_{\rho_n}}, P_{(-\infty,\log L_n)} (\log \Delta_{\rho_n|\sigma_n})\, \Omega_{\rho_n}\big\rangle.
	\end{align*}
	As we remarked after definition \eqref{eq_defDelta}, $J \Delta_{\rho|\sigma} J = \Delta_{\sigma|\rho}^{-1}$, with $J$ being the map $:A\mapsto A^*$. This, together with the fact that $J^2=\id$ implies that $\log \Delta_{\rho_n|\sigma_n}=-J (\log \Delta_{\sigma_n|\rho_n}) J$, and we obtain
	\[P_{(-\infty,\log L_n)} (\log \Delta_{\rho_n|\sigma_n}) = P_{(-\infty,\log L_n)} (-J\log \Delta_{\sigma_n|\rho_n} J)= J  P_{(-\infty,\log L_n)} (-\log \Delta_{\sigma_n|\rho_n})\,J.\]
Noting that $J$ is anti self-adjoint \eqref{antiJ}, $J \Omega_{\rho_n} = \Omega_{\rho_n}$, and the definition of the random variable $X_n\sim\mu_{\sigma_n|\rho_n}$, we obtain
	\begin{align*}
	\alpha(T_n) &\leq   \big\langle{\Omega_{\rho_n}}, P_{(-\infty,\log L_n)} (\log \Delta_{\rho_n|\sigma_n})\, \Omega_{\rho_n}\big\rangle\\
	&=   \big\langle{\Omega_{\rho_n}}, J P_{(-\infty,\log L_n)} (-\log \Delta_{\sigma_n|\rho_n})J\, \Omega_{\rho_n}\big\rangle\\
    &=\langle J\Omega_{\rho_n},P_{(-\infty, \log L_n)}(-\log \Delta_{\sigma_n|\rho_n})J\Omega_{\rho_n}\rangle\\
	&=\langle \Omega_{\rho_n}, P_{(-\infty,\log L_n)} (-\log \Delta_{\sigma_n|\rho_n})\Omega_{\rho_n}\rangle\\
	&=\PP(X_n\leq \log L_n)=\mathbb{P}\left(\frac{X_n-D(\rho_n||\sigma_n)}{\sqrt{w_n\,}}\leq t_2+\frac{f(w_n)}{\sqrt{w_n\,}}\right).
	\end{align*}
	By consequence \eqref{convg} of Bryc's theorem, the above converges to $\Phi(t_2/\sqrt{v(\hat \rho,\hat{\sigma})})<\varepsilon$. Therefore, for $n$ large enough, one has $\alpha(T_n)\leq \varepsilon$. The bound \eqref{eq_betan} proves \Cref{eq_maineq2}.

\qed

\section{Examples} \label{sec_Examples}

In this section we give examples of sequences of quantum states, $\hat{\rho} = (\rho_n)_{n \in \NN}$ and $\hat{\sigma} = (\sigma_n)_{n \in \NN}$, for which \Cref{cond-Bryc} holds, and hence \Cref{mainresult} applies. Our examples closely follow those from \cite{Jaksicetal}, since \Cref{cond-Bryc} is stronger than the conditions required for the (first order) Stein lemma given in that reference.

\subsection{Quantum i.i.d.~states}
The simplest example is that of i.i.d.~states: $\rho_n = \rho^{\otimes n}$ and 
$\sigma_n= \sigma^{\otimes n}$, with $\rho, \sigma \in \cD(\cH)$ and $\cH_n = \cH^{\otimes n}$. 
For such states, from the definitions \reff{qrel} and \reff{info-var} of the quantum relative entropy and the quantum information variance it follows that
\begin{equation*}
D(\rho_n || \sigma_n) = n D(\rho||\sigma) \quad {\hbox{and}} \quad V(\rho_n || \sigma_n) = n V(\rho||\sigma)\label{nv}.
\end{equation*}
\noindent Hence, choosing $w_n =n$ in \Cref{drho}, we obtain
\begin{align*}
d(\hrho, \hsigma) = D(\rho||\sigma) \qquad  v(\hrho, \hsigma) = V(\rho||\sigma).
\end{align*}
Moreover, $\Psi_{1-z}(\rho_n | \sigma_n) = n\Psi_{1-z}(\rho|\sigma)$. 

From the above it follows that \Cref{cond-Bryc} is satisfied. Hence, in this case \Cref{mainresult} holds, and more precisely \Cref{coro_ideal} holds. This result was obtained independently by Tomamichel and Hayashi \cite{TH13} and Li \cite{L14}, employing the Central Limit Theorem (CLT). The fact that it follows from our analysis is not surprising, since for i.i.d.~states Bryc's theorem reduces to the CLT.
In addition, in the i.i.d.~case we have the more powerful Berry-Esseen theorem (see {e.g.~\cite{Feller2}) which can be used to obtain information about the behaviour of the third order term 
in the expansion of $- \log \beta_n(\eps)$, as stated in the following theorem.
\begin{theorem}
	Fix $\eps\in(0,1)$. Then in the case in which the states $\rho_n$ and $\sigma_n$ in the sequences $\hrho$ and $\hsigma$ are i.i.d.~states on the sequence of finite dimensional Hilbert spaces $\cH^{\otimes n}$, i.e.~$\rho_n=\rho^{\otimes n}$ and $\sigma_n=\sigma^{\otimes n}$, we have 
	$$-\log\beta_n(\eps)=nD(\rho||\sigma)+\sqrt{nV(\rho||\sigma)}\,\Phi^{-1}(\eps)+\mathcal{O}(\log n).$$
\end{theorem}
\smallskip

\noindent
\begin{proof} We first prove the bound
\begin{align}
-\log\beta_n(\eps) &\leq nD(\rho||\sigma)+\sqrt{nV(\rho||\sigma)}\,\Phi^{-1}(\eps)+\mathcal{O}(\log n)
\end{align}
Following the proof of \Cref{eq_maineq1} of \Cref{mainresult} up to \Cref{berryesseenchanges} (with the choice $w_n = n$), we infer that for any test $T_n$ such that $\tr \left((\mathbb{I}_n-T_n)\rho_n\right)\le \eps$, for any $\theta_n, v_n \in {\mathbb{R}}$
	\begin{align}\label{berryx}
	\tr (T_n\sigma_n)\ge \operatorname{e}^{-\theta_n}\Big( \frac{\mathbb{P}(X_n\le v_n)}{1+\operatorname{e}^{v_n-\theta_n}}  -\eps\Big),
	\end{align}
where $X_n\sim \mu_{\sigma_n|\rho_n}$ is a sum of $n$ $i.i.d.$ random variables of law $\mu_{\sigma|\rho}$ (compute e.g. its characteristic function). By the Berry-Esseen theorem, for any real $a$,  
${\mathbb{P}}\left(\frac{X_n -n D(\rho||\sigma)}{\sqrt{nV(\rho||\sigma)}} \le a\right)
-\Phi(a)=\mathcal{O}(n^{-1/2})$. 
Choosing $v_n=nD(\rho||\sigma)+\sqrt {n V(\rho||\sigma)}\, \Phi^{-1}(\eps)$,
in \Cref{berryx}, gives
\begin{align*}
	\tr (T_n \sigma_n)&\ge \operatorname{e}^{-\theta_n}\Big( \frac{\eps +\mathcal{O}(n^{-1/2})}{1+\operatorname{e}^{v_n-\theta_n}}-\eps\Big)
\end{align*}
Further, choosing $\theta_n=v_n+\log\sqrt{n}$, we obtain
\begin{align*}
	\tr (T_n\sigma_n)\ge \operatorname{e}^{-nD(\rho||\sigma)-\sqrt{nV(\rho||\sigma)}\,\Phi^{-1}(\eps)-\log\sqrt{n}}\Big( \frac{\mathcal{O}(n^{-1/2})-\eps n^{-1/2}}{1+n^{-1/2}} \Big),
\end{align*}	
so that
$$-\log \tr(\sigma_nT_n)\le nD(\rho||\sigma)+\sqrt{nV(\rho||\sigma)}\,\Phi^{-1}(\eps) +\mathcal{O}(\log n).$$
\medskip

\noindent
Next, to prove the opposite inequality, namely,
$$-\log \tr(\sigma_nT_n)\le nD(\rho||\sigma)+\sqrt{nV(\rho||\sigma)}\,\Phi^{-1}(\eps) +\mathcal{O}(\log n),$$
we proceed similarly to the proof of inequality \reff{eq_maineq2} of \Cref{mainresult} (with the choice $w_n = n$), except that we replace $o(\sqrt{w_n})$
in the expression of $L_n$ with $\cO(\log n)$.			\qed
\end{proof}

\subsection{Quantum Spin Systems} \label{sec_QSS}
In this section we consider the example of quantum spin systems on the lattice $\Lambda=\ZZ^d$ (all results below hold for any lattice of finite degree). Here we closely follow the discussion and notation of \cite{Jaksicetal}. Let $\cX$ denote the set of all finite subsets of $\Lambda$, and 
$\cX_0 \subset \cX$ the set of subsets of $\Lambda$ containing $0$.
At each lattice site $x\in\Lambda$, there is a particle with a finite 
number of internal (or spin) degrees of freedom, and we assume that each particle has an identical state space~$\mathfrak h$. For bookkeeping purposes we denote by $\h_x$ the copy of $\h$ associated with the particle at  site~$x$. For $X\in \cX$ we let $\h_X=\bigotimes_{x\in X} \h_x$ and for $X_1$, $X_2\in\cX$ with $X_1\subset X_2$ we identify $\h_{X_1}$ with a subspace of $\h_{X_2}$ via the decomposition $\h_{X_2}= \h_{X_1}\otimes \h_{X_2\setminus X_1}$.

Consider the algebras $\cA_X$ of all bounded operators acting on $\h_X$, $X\in\mathcal{X}$, with the usual operator norm and with Hermitian conjugation as the $*$-involution. The algebras $\cA_X$ can be considered to be partially nested, i.e.
$$\cA_{X_1} \subset \cA_{X_2} \quad{\hbox{if}} \,\, X_1 \subset X_2,$$
by identifying each operator $A_1 \in \cA_{X_1}$ with the operator $A_1 \otimes \mathbb{I}_{X_2 \setminus X_1} \in \cA_{X_2}$, where $\mathbb{I}$ denotes the identity operator. Moreover, the algebras $\cA_X$ are {\em{local}},  i.e.~if $A_1 \in \cA_{X_1}$ and $A_2 \in \cA_{X_2}$,
and $X_1 \cap X_2 = \emptyset$, then $A_1A_2 = A_2A_1$. The norm closure of $\bigcup_{X \nearrow \ZZ^d} \cA_X$ defines an algebra which we denote by $\cA$. It is the quasilocal C*-algebra of observables associated with the infinite lattice $\ZZ^d$. All local algebras $\cA_X$ are subalgebras of $\cA$, and we denote by $\|A\|$ the norm of any operator $A\in \cA$.

For any $X$ in $\cX$ we denote by $|X|$ its cardinality, by $\mathrm{diam}(X)$ its diameter 
$$\mathrm{diam}(X):= \max\{|x-y|_{\ZZ^d} \ \mathrm{s.t.} \  x,y\in X\},$$
where $|x-y|_{\ZZ^d}:= \sum_{i=1}^d|x_i - y_i|$ denotes the ``Manhattan distance'' between the two sites $x, y\in \ZZ^d$; here $x=(x_1, \ldots, x_d)$ and $y=(y_1, \ldots, y_d)$. For 
$a\in \Lambda$ we define a map $T_a$ as the identity $\h_x\to\h_{x+a}$ which we extend as a map $\h_X\to\h_{X+a}$ where $X+a=\{x+a\, |\, x\in X\}$. 
The group of space translations ${\mathbf{T}}_{\ZZ^d}$ acts as a $*$-automorphism group $\{T_a \,:\, a \in \ZZ^d\}$ on $\cA$, and for any $X \subset \ZZ^d$
$$ \cA_{X + a} = T_a \cA_X T_{-a}.$$

An {\em{interaction}} of a quantum spin system is a function $\Phi$  from finite, nonempty subsets~$X$ of $\ZZ^d$, to self-adjoint observables $\Phi_X \in \cA_X$. The interactions are said to be translation invariant if
$$ T_a \Phi_X T_{-a}= \Phi_{X+a} \quad \mbox{for all}\ a \in \ZZ^d, \,\, X \subset \ZZ^d.$$
The following quantity will be relevant for us:
\begin{equation}
\|\Phi\| = \sum_{X\in \cX_0} \|\Phi_X\|. 
\end{equation}
We say that the interaction has finite range if there exists an $R>0$ such that $\mathrm{diam}\, X \geq R$ implies $\Phi(X)=0$; the smallest such $R$ is called the range of $\Phi$. Notice that an interaction $\Phi$ with finite range has finite $\|\Phi\|$. 
For any $X\in \cX$ we associate to an interaction~$\Phi$ the so-called {\em{interaction Hamiltonian}} on $X$:
for two fixed interactions $\Phi$ and $\Psi$ we define the following Hamiltonians:
\begin{equation} \label{eq_defHamQSS}
H_X^\Phi = \sum_{Y\subset X} \Phi_Y \quad {\hbox{and}} \quad H_X^\Psi = \sum_{Y\subset X} \Psi_Y.
\end{equation}

We are interested in {\em{finite-volume Gibbs states}} associated with $\Phi$ and $\Psi$ respectively. 
More precisely, for every $n\in\NN$ we let $\Lambda_n=\{-n,\ldots,+n\}^d$ and define two Gibbs states on $\cH_n:={\h}_{\Lambda_n}$ by density matrices
\begin{equation}\label{eq_defGibbsstates}
\rho_n= \frac{\e^{-\beta_1 H_{\Lambda_n}^\Phi}}{\tr(\e^{-\beta_1 H_{\Lambda_n}^\Phi})} \quad {\hbox{and}} \quad \sigma_n= \frac{\e^{-\beta_2 H_{\Lambda_n}^\Psi}}{\tr(\e^{-\beta_2 H_{\Lambda_n}^\Psi})}
\end{equation}
where $\beta_1$, $\beta_2$ $\in(0,\infty)$ are \textit{inverse temperatures} (the physical definition would actually require $\beta=(k_B T)^{-1}$ where $T$ is the temperature and $k_B$ is Boltzmann's constant; here we set $k_B = 1$). Using the results of \cite{NR} we can prove the following:
\begin{proposition}\label{prop-qss}
If $\Phi$ and $\Psi$ are translation invariant, finite range interactions, then for high enough temperatures (i.e.~for inverse temperatures $\beta_1$ and $\beta_2$ small enough), the sequence $(\rho_n,\sigma_n)_{n\in\NN}$ of pairs of finite volume Gibbs states on $\Lambda_n = \{-n, \ldots, n\}^d$ (defined through \Cref{eq_defGibbsstates} above),  satisfies \Cref{cond-Bryc} with respect to the weights $w_n=|\Lambda_n|=(2n +1)^d$.
\end{proposition}

\noindent\begin{proof}
In the present context, the quantity 
$E_n(z)$ appearing in \Cref{cond-Bryc} is given by
\begin{align}
E_n(z) &:= \Psi_{1-z}(\rho_n |\sigma_n)\nonumber\\
&=\log \tr (\e^{-(1-z) \beta_1 H_{\Lambda_n}^\Phi} \e^{-z \beta_2 H_{\Lambda_n}^\Psi})  - (1-z) \log \tr (\e^{- \beta_1 H_{\Lambda_n}^\Phi})-z \log \tr (\e^{-\beta_2  H_{\Lambda_n}^\Psi}).
\end{align}
It is a well-known result (see e.g.~\cite{Simon}) that, under the assumptions of \Cref{prop-qss}, the quantities
\[\frac{1}{|\Lambda_n|} \log \tr (\e^{- \beta_1 H_{\Lambda_n}^\Phi}) \quad {\hbox{and}} \quad \frac{1}{|\Lambda_n|} \log \tr (\e^{-\beta_2  H_{\Lambda_n}^\Psi})\]
converge as $n\to\infty$, the limits being called the \textit{pressure} associated with the pairs $(\beta_1,\Phi)$ and $(\beta_2, \Psi)$, respectively. 
Hence, to establish \Cref{cond-Bryc} it suffices to prove that there exists an $r >0$ such that the functions
\begin{align}
f_n(z) &:= \frac1{|\Lambda_n|} \log \tr (\e^{-(1-z)\beta_1 H_{\Lambda_n}^\Phi} \, \e^{-z \beta_2 H_{\Lambda_n}^\Psi})
\label{fnz}
\end{align}
are analytic on $B_\CC(0,r)$, converge pointwise on the real segment $(-r,+r)$, and admit a uniform (in $n$ and $z$) bound.

Our proof relies on a result of \cite{NR}, namely Proposition 7.10 of the paper, which in the present context can be stated as \Cref{prop:NR} below. Before stating it, however, we need to define certain quantities.

Consider the completely mixed state acting on the Hilbert space $\cH_n := \h_{\Lambda_n}$. It is an element of the observable algebra $\cA_{\Lambda_n}$ and is given by the tensor product (or i.i.d.) density matrix
$$ \tau_n := \bigotimes_{x \in \Lambda_n} \frac{I_x}{ {\rm{dim}}\, \h} 
= \frac{1}{({\rm{dim}}\, \h)^{|\Lambda_n|}}  \bigotimes_{x \in \Lambda_n} {I_x},
$$
where $I_x \in \cA_x$ denotes the identity operator acting on $\h_x$, and 
$ {\rm{dim}}\, \h$ is the dimension of $\h$ (recall $\h_x \simeq \h$).

Further, define the following quantity, for any $z_1, z_2 \in \mathbb{C}$:
$$ Z_{\Lambda_n}^{z_1,z_2} := \tr \left( \tau_n e^{z_1 \beta_1 H_{\Lambda_n}^\Phi}
e^{z_2 \beta_2 H_{\Lambda_n}^\Psi}\right).$$
Its logarithm is hence given by
\begin{align}\label{log-z}
\log Z_{\Lambda_n}^{z_1,z_2} &= - |\Lambda_n| \log ({{\dim}}\, \h) + \log \tr\left( e^{z_1 \beta_1 H_{\Lambda_n}^\Phi}e^{z_2 \beta_2 H_{\Lambda_n}^\Psi}  \right).
\end{align}
Note in particular, that for the choice $z_1 = z-1$ and $z_2= -z$, the function $f_n(z)$ (defined through \Cref{fnz}), whose analyticity 
and convergence properties we
are interested in, can be expressed in terms of the above quantity as follows:
\begin{align}\label{reln}
f_n(z) = \frac{1}{|\Lambda_n|} \log Z_{\Lambda_n}^{z-1, -z} +  \log ({{\dim}}\, \h).
\end{align}
Hence, we can deduce the desired properties of $f_n(z)$ from those of $\log Z_{\Lambda_n}^{z-1,-z}$.

The quantity $\log Z_{\Lambda_n}^{z_1,z_2}$ has an expansion of the form (see \cite{NR})
\begin{align}\label{cluster}
 \log Z_{\Lambda_n}^{z_1,z_2} &= \sum_{C \subseteq \Lambda_n} \rw^{z_1,z_2}(C),
\end{align}
where the {\em{cluster weights}} $\rw^{z_1,z_2}(C)$ are also translation-invariant, i.e.~\[\rw^{z_1,z_2}(C)=\rw^{z_1,z_2}(C+x)\]
 for any $x\in \Lambda$, $C\in\mathcal X$. Such an expansion is called a {\em{cluster expansion}}. Proposition 7.10 of~\cite{NR} is a statement of analyticity of these cluster weights and the boundedness of the cluster expansion. An immediate simplification can be stated as follows:
\begin{proposition}\label{prop:NR}
Assume that $\Phi$ and $\Psi$ are translation-invariant interactions, and 
\begin{align}\label{c1}
\sum_{X \in \cX_0} e^{2a|X|} \left(e^{\delta(\beta_1 ||\Phi_X|| + \beta_2||\Psi_X||)} -1 \right) \leq a,
\end{align}
for some $a, \delta>0$. 
Then all cluster weights are analytic in 
$\cD:=B_\CC(0,\delta)^2$
and 
\begin{align}
\sup_{(z_1, z_2) \in \cD} \sum_{C_0\in \cX_0} | {\rw}^{z_1,z_2}(C_0)| \leq a
\end{align}
\end{proposition}
To employ the above proposition in our proof, we proceed as follows. Denote by $R$ the range of $\Phi$,
and let $q=\max \{\beta_1 ||\Phi||, \beta_2||\Psi||\}$.
We fix $a >0$ and $\delta =1+r>1 $. The number of sets $X$ that contain the origin $0$, and have diameter at most $R$ is bounded by $2^{N_{R,d}}$ with $N_{R,d}=(2R+1)^d$. A crude upper bound for the left hand side of \Cref{c1} is 
$$ 2^{N_{R,d}} \, e^{2a N_{R,d}} \left( e^{2 \delta  q} - 1\right).$$
Since this upper bound is zero for $q=0$, it remains smaller than $a$ for $q$ small but nonzero. 
For this value of $q$, \eqref{c1} holds, so that the claim of the proposition holds with $\mathcal D=B_\CC(0,1+r)^2$. In particular, setting $z_1 = z-1$ and $z_2= -z$, and using the translation invariance of the interaction $\Phi$ and $\Psi$, we infer that
\begin{equation} \label{eq_uniformboundclusterexpansion}
\sup_{z\in B_{\mathbb{C}}(0,r)} \sum_{C_0\in\cX_0} |\rw^{z-1, -z}(C_0)|\leq a.
\end{equation}
Now observe that any $C \in \cX$ can be written $x+C_0$ with $x\in \Lambda$ and $C_0\in\cX_0$ in exactly $|C|=|C_0|$ distinct ways. Then, using the 
notation $P_c$, which is equal to $1$ if the condition $c$ holds and 
$0$ otherwise, we have
\begin{align*}
\sum_{C\subset\Lambda_n} \rw^{z-1, -z}(C)&= \sum_{x\in \Lambda_n} \sum_{C_0\in \cX_0} 
P_{x+C_0\subseteq \Lambda_n} \frac{\rw^{z-1, -z}(C_0)}{|C_0|}\\
|\Lambda_n| \sum_{C_0\in\cX_0} \frac{\rw^{z-1, -z}(C_0)}{|C_0|} &= \sum_{x\in \Lambda_n} \sum_{C_0\in \cX_0} \frac{\rw^{z-1, -z}(C_0)}{|C_0|},
\end{align*}
so that by \eqref{reln},
\[\sum_{C_0\in\cX_0} \frac{\rw^{z-1, -z}(C_0)}{|C_0|} -\frac1{|\Lambda_n|}\sum_{C\subseteq\Lambda_n} \rw^{z-1, -z}(C) =  \frac1{|\Lambda_n|}\sum_{x\in \Lambda_n} \sum_{C_0\in \cX_0} P_{x+C_0\not\subseteq \Lambda_n} \frac{\rw^{z-1, -z}(C_0)}{|C_0|}.\]
By \Cref{eq_uniformboundclusterexpansion}, one has
\[ \Big|\sum_{C_0\in\cX_0} \frac{\rw^{z-1, -z}(C_0)}{|C_0|}\Big|\leq  \sum_{C_0\in\cX_0} \frac{|\rw^{z-1, -z}(C_0)|}{|C_0|} \leq \sum_{C_0\in\cX_0} |\rw^{z-1, -z}(C_0)|\leq a\]
for any $z\in B_\CC(0,r)$. Now, by \Cref{reln,cluster}, \Cref{cond-Bryc} will follow for this model if we can prove
\begin{equation}\label{eq_cond1QSS}
\lim_{n\to\infty} \sup_{z\in B_\CC(0,r)} \Big|\frac1{|\Lambda_n|}\sum_{x\in \Lambda_n} \sum_{C_0\in \cX_0} P_{x+C_0\not\subseteq \Lambda_n} \frac{\rw^{z-1, -z}(C_0)}{|C_0|}\Big| =0.
\end{equation}
Fix $z\in B_\CC(0,r)$. From \eqref{eq_uniformboundclusterexpansion}, for any $\varepsilon>0$, there exists a finite set $\Sigma\subset\cX_0$ such that 
\[ \sum_{C_0\in \cX_0\setminus\Sigma} |\rw^{z-1,-z}(C_0)| < \varepsilon.\]
Denote by $S=\max \{\mathrm{diam}\,C_0\,|\, C_0\in \Sigma\}$, and let 
\[\Lambda_n^\Sigma = \{x\in\Lambda_n \, |\, x+C_0 \subset \Lambda_n \, \forall C_0\in\Sigma\}.\]
Obviously for $n>S$ one has $\Lambda_{n-S}\subset \Lambda_n^\Sigma$, and
\begin{align*}
\Big|\sum_{x\in \Lambda_n} \sum_{C_0\in \cX_0} P_{x+C_0\not\subseteq \Lambda_n}\frac{\rw^{z-1, -z}(C_0)}{|C_0|}\Big| &\leq \sum_{x\in \Lambda_n} \sum_{C_0\in \cX_0} 
P_{x+C_0\not\subseteq \Lambda_n}\frac{|\rw^{z-1, -z}(C_0)|}{|C_0|}\\
&\leq\sum_{x\in \Lambda_n^\Sigma} \sum_{C_0\in \cX_0} P_{x+C_0\not\subseteq \Lambda_n} \frac{|\rw^{z-1, -z}(C_0)|}{|C_0|}\\
&\qquad + \sum_{x\in \Lambda_n\setminus\Lambda_n^\Sigma} \sum_{C_0\in \cX_0} P_{x+C_0\not\subseteq \Lambda_n} \frac{|\rw^{z-1, -z}(C_0)|}{|C_0|}.
\end{align*}
By definition of $\Lambda_n^\Sigma$ one has
\[  \sum_{x\in \Lambda_n^\Sigma} \sum_{C_0\in \cX_0} P_{x+C_0\not\subseteq \Lambda_n} \frac{|\rw^{z-1, -z}(C_0)|}{|C_0|} \leq|\Lambda_n^\Sigma| \sum_{C_0\in \cX_0\setminus\Sigma} \big|\rw^{z-1, -z}(C_0)\big|<\varepsilon \,|\Lambda_n^\Sigma|. \]
On the other hand,
\[\sum_{x\in \Lambda_n\setminus\Lambda_n^\Sigma} \sum_{C_0\in \cX_0} P_{x+C_0\not\subseteq \Lambda_n} \frac{|\rw^{z-1, -z}(C_0)|}{|C_0|}\leq  | \Lambda_n\setminus\Lambda_n^\Sigma|\sum_{C_0\in \cX_0}  |\rw^{z-1, -z}(C_0)| \leq a\,| \Lambda_n\setminus\Lambda_n^\Sigma|. \]
This implies immediately 
\[ \sup_{z\in B_\CC(0,r)} \Big|\frac1{|\Lambda_n|}\sum_{x\in \Lambda_n} \sum_{C_0\in \cX_0} P_{x+C_0\not\subseteq \Lambda_n} \frac{\rw^{z-1, -z}(C_0)}{|C_0|}\Big| \leq \varepsilon + a\,\Big(1-\Big(\frac{2(n-S)+1}{2n+1}\Big)^d\Big),\]
where we use the fact that $|\Lambda_{n-S}| \leq |\Lambda_n^\Sigma|$.
Since $\varepsilon$ is arbitrary, this implies \Cref{eq_cond1QSS} and therefore \Cref{cond-Bryc}. 

\qed
\end{proof}

\subsection{Free fermions on a lattice} \label{sec_CAR}
In this section we consider quasi-free states of a fermionic lattice gas. We closely follow the treatment of \cite{MHOF}. This is given, for example, by states of non-interacting electrons on the lattice $\ZZ^d$, which are allowed
to hop from site to site. Since the electrons are fermions, they are subject to Fermi-Dirac statistics and hence they have to obey the Pauli exclusion principle. The prefix ``quasi''
in ``quasi-free'' can be heuristically understood to arise from the fact that, even though the electrons do not interact among themselves, their movement on the lattice is constrained
by the Fermi-Dirac statistics.

Let us give a brief description of the setup. Let us start with the one-particle Hilbert space $\h := \ell^2(\ZZ^d)$. It represents a single particle which is confined to the lattice $\Lambda:=\ZZ^d$. To incorporate the statistics of the particle, we consider the Fock space
$\cF(\h):= \oplus_{n \in {\mathbb{N}}} \wedge^n \h$, with the convention $\wedge^0= \mathbb{C}$. Here we use the notation
$$x_1 \wedge x_2 \ldots \wedge x_n = \frac{1}{\sqrt{n!}}\sum_{\sigma \in S_n} {\rm{sgn}}(\sigma) x_{\sigma(1)} \otimes x_{\sigma(2)} \ldots \otimes x_{\sigma(n)}$$
where the summation runs over all permutations of $n$ elements and ${\rm{sgn}}(\sigma)$ denotes the
sign of the permutation $\sigma$. The algebra of observables is the C*-algebra generated by the creation and annihilation operators on the Fock space, obeying the canonical anticommutation relations (CAR). More precisely, for each $x \in \h$,
the creation operator $a^*(x)$ is the unique bounded linear extension $a^*(x): \cF(\h) \mapsto \cF(\h)$ of 
$$
a^*(y) : x_1  \wedge \ldots \wedge x_n \mapsto y \wedge x_1  \wedge \ldots \wedge x_n, \quad x_1, \ldots, x_n \in \h, \, n \in {\mathbb{N}},$$
and the corresponding annihilation operator is its adjoint $a(x) := (a^*(x))^*$.
Creation
and annihilation operators satisfy the canonical anticommutation relations (CAR):
$$a(x)a(y) + a(y)a(x) = 0 , \quad a(x)a^*(y) + a^*(y) a(x) = \langle x, y \rangle \,{\mathbb{I}}$$
The C*-algebra generated by $\{ a(x) : x \in \h \}$ is called the algebra of the canonial
anticommutation relations, and is denoted by ${\rm{CAR}}(\h)$. The number operator $N$ is defined in terms of the creation and 
annihilation operators as follows: $N = \sum_k a_k^* a_k$ where $a_k = a(e_k)$ for some orthonormal basis $\{e_k\}$ of $\h$. The algebra $\mathrm{CAR}(\h)$ is infinite-dimensional.
However, the objects associated with $\h_n=\ell^2(\Lambda_n)$, where $\Lambda_n$ is a finite subset of the lattice $\ZZ^d$, are finite-dimensional.

Let us denote the standard basis in $\ell^2(\ZZ^d)$ by $\{e_{\bfi} : \bfi \in \ZZ^d\}$. Then {\em{shift operators}} are
defined as the unique linear extensions of $T_\bfj : e_\bfi \to e_{\bfi + \bfj}$, $\bfi \in \ZZ^d$, $\forall\, \bfj \in \ZZ^d$. An operator which commutes with all 
the unitaries $T_\bfj$ is said to be shift-invariant. Shift-invariant operators on $l^2(\ZZ^d)$ commute with each other. Moreover, defining the Fourier transformation
$$
{\mathfrak{F}}:\,\ell^2(\ZZ^d)\to L^2([0,2\pi)^d);\, \quad \, {\mathfrak{F}} e_{\bfk}:=\varphi_{\bfk}(\bfx), \quad {\hbox{where}} \,\, \varphi_{\bfk}(\bfx) = e^{i\langle \bfk, \bfx \rangle}, \,\, \bfx \in [0,2\pi)^d, \,\, \bfk \in \ZZ^d,
$$
where $\langle \bfk, \bfx \rangle:=\sum_{i=1}^d k_ix_i$, every shift-invariant operator $Q$ can be expressed in the form $Q=\mathfrak{F}^{-1}M_{\hq}\mathfrak{F}$, where $M_{\hq}$ denotes the multiplication operator by a bounded measurable function $\hq$ on $[0,2\pi)^d$. The operator of $Q$ in the basis $\{e_{\bfi} : \bfi \in \ZZ^d\}$ is a Toeplitz operator, and its elements are given by
$$Q_{\bfj, \bfk} \equiv \langle e_\bfj, Q e_\bfk\rangle = \frac{1}{(2\pi)^d} \int e^{- i \langle\bfj - \bfk, \bfx\rangle} \hq(\bfx) d\bfx.$$

Let $H \in \cB(\h)$ denote the one-particle Hamiltonian of a system of free fermions on the lattice $\Lambda$. Then an example of a quasi-free state of a fermionic lattice gas is its Gibbs state, which for an inverse temperature $\beta$, is given by the density matrix 
$\rho_\beta := e^{-\beta \hat{H}}/{\tr(e^{-\beta\hat{ H}})}$, where $\hat{H}=\d\Gamma(H)$ denotes the (differential) second quantization of $H$, acting on the Fock space $\cF(\h)$ (see \cite{BR2} or Merkli's notes in \cite{OQS1}). Consider, in particular, examples in which the Hamiltonian is the second quantization of a shift-invariant operator, implying that the total number of fermions is conserved. The Gibbs state can be expressed as a linear form on ${\rm{CAR}}(\h)$ as follows 
\begin{align}\label{QomegaQ}
	{\omega_Q\Big(a^*(x_1)\ldots a^*(x_n) a(y_m) \ldots a(y_1)\Big)}&:= \tr \big( \rho_\beta\, a^*(x_1)\ldots a^*(x_n) a(y_m) \ldots a(y_1)\big)\nonumber\\
	&= \delta_{mn} \,{\rm{det}} \big(\langle y_i, Q x_j\rangle\big)_{i,j},
\end{align}
where $Q=\e^{-\beta H}(1+\e^{-\beta H})^{-1} \in \cB(\h)$, where $H$ denotes the one-particle Hamiltonian. 
The state $\omega_Q$ is in then called a quasi-free state with {\em{symbol}} $Q$. It is shift-invariant, since the underlying Hamiltonian is chosen to be shift-invariant.

For our quantum hypothesis testing problem we proceed as follows. Consider a sequence of finite subsets $\Lambda_n := \{0,\ldots,n-1\}^d \subset \ZZ^d$ and the subspaces $\h_n=\ell^2(\Lambda_n)$ which we can view as subspaces of $\h$.  Let $\omega_Q$ and $\omega_R$ be two shift-invariant Gibbs states with symbols $Q$ and $R$ respectively, such that $\delta < Q, R < 1- \delta$ for some 
$\delta \in (0,1/2)$. By the discussion above, 
the operators ${\mathfrak{F}} Q {\mathfrak{F}}^{-1}$ and ${\mathfrak{F}} R {\mathfrak{F}}^{-1}$ on $L^2([0,2\pi)^d)$ are multiplication operators by bounded measurable functions $\hq, \hr : [0, 2 \pi)^d \to [0,1]$ with essential range in $[\delta,1-\delta]$. This condition ensures that the ``local restrictions" of $\omega_Q$, $\omega_R$ which we now define, are faithful. Let $Q_n=P_n Q P_n$, $R_n=P_n R P_n$, where $P_n$ denotes the orthogonal projection operator on~$\h_n$. Then we choose the sequence of states $(\rho_n, \sigma_n)_{n \in {\mathbb{N}}}$ of our hypothesis testing problem to be given by density matrices corresponding to states $\omega_{Q_n}$ and $\omega_{R_n}$ on $\mathrm{CAR}(\h_n)$, associated with the symbols $Q_n$, $R_n$ respectively  through \Cref{QomegaQ}. That is, the null-hypothesis in this case is that the true state of the infinite system is $\omega_Q$, while the alternative hypothesis is that it is $\omega_R$, and we make measurements on local subsystems to decide between these two options. By Lemma 3 in \cite{Dierckx2008a}, the density matrices $\rho_n$, $\sigma_n$ associated with $\omega_{Q_n}$ and $\omega_{R_n}$ are
\begin{align}\label{gibbs}
	\rho_n=\det(\mathbb I - Q_n) \, \bigoplus_{k=0}^{\dim \mathfrak h _n} \bigwedge^k \frac{Q_n}{\mathbb I - Q_n} \qquad 
	\sigma_n=\det(\mathbb I - R_n) \, \bigoplus_{k=0}^{\dim \mathfrak h _n} \bigwedge^k \frac{R_n}{\mathbb I - R_n}
\end{align}
where for any operator $A$ on $\mathfrak h$, $\wedge^k A$ denotes the restriction of $A^{\otimes k}$ to the antisymmetric subspace $\wedge^k \h$ of $\mathfrak h^{\otimes k}$. 
As shown in \cite{MHOF}, it follows from \Cref{gibbs} that for any $s\in\mathbb{R}$
\begin{align}\label{psisfermions}
	\Psi_s(\rho_n|\sigma_n)=\tr \log (\mathbb{I}_n-Q_n)^{s}+
	\tr \log(\mathbb{I}-R_n)^{1-s} + \tr \log (\mathbb{I}_n+W_n^{(s)}),
\end{align}	
where $$W_{n}^{(s)}=\left(\frac{Q_n}{\mathbb{I}_n-Q_n} \right)^{\frac{s}{2}}   
\left(\frac{R_n}{\mathbb{I}_n-R_n} \right)^{1-s}    \left(\frac{Q_n}{\mathbb{I}_n-Q_n} \right)^{\frac{s}{2}}.$$
and
\begin{align} \label{eq_EntRevCAR}
	D(\rho_n ||\sigma_n) &=
	\frac{\d}{\d s}\Psi_s(\rho_n|\sigma_n)\vert_{s=1}\nonumber\\
	&= \tr (Q_n \log Q_n) +\tr  ((\mathbb I_n-Q_n)\log(\II_n-Q_n)) \nonumber \\&\qquad-\tr (Q_n\log R_n)-\tr ((\II_n-Q_n) \log(\II_n-R_n)).
\end{align}

\begin{proposition}
	For the sequence $(\rho_n,\sigma_n)_{n\in\NN}$ defined above, we have
	\begin{align} 
		d(\hrho,\hsigma) &:= \lim_{n \to \infty} \frac1{n^d}\,D(\rho_n||\sigma_n) \\
		&=\frac1{(2\pi)^d}\int_{[0,2\pi)^d}\hq(\bfx)\log \frac{\hq(\bfx)}{\hr(\bfx)} +(1-\hq(\bfx)) \log \frac{1-\hq(\bfx)}{1-\hr(\bfx)}\,{\mathrm{d}} \bfx.\label{eq_CARD}
	\end{align}
	and for any $s\in\mathbb{R}$
	\begin{align}\label{eq_psi}
		\lim_{n \to \infty} \frac{1}{n^d} \Psi_s(\rho_n|\sigma_n)
		&=\frac1{(2\pi)^d}\int_{[0,2\pi)^d}\log \big(\hq^{s}(\bfx)\hr^{1-s}(\bfx)+(1-\hq(\bfx))^s(1-\hr(\bfx))^{1-s} \big){\mathrm{d}} \bfx.
	\end{align}
	Moreover the sequence $(\rho_n,\sigma_n)_{n\in\NN}$ satisfies \Cref{cond-Bryc} with respect to the weights $w_n=|\Lambda_n|$.
	Assuming moreover that $Q$ satisfies 
	\begin{align}\label{condsoncoroll}
		\sum_{\bfk \in \ZZ^d}|\bfk|^{d}\left|\langle  e_\bfk,Qe_\mathbf{0}\rangle\right|<\infty,
	\end{align}	
	then
	\begin{align}\label{strong}
		D(\rho_n||\sigma_n)-n^d\, d(\hrho,\hsigma) = \mathcal O(n^{d-1}).
	\end{align}
\end{proposition}
\begin{remark}
	This proves that whenever \Cref{condsoncoroll} is satisfied, then the assumptions of \Cref{coro_ideal} hold for $d = 1$. In the general case, denote $C$ such that $D(\rho_n||\sigma_n)=n^2\, d(\hrho,\hsigma) + n \, C + O(1)$, then for $d=2$, the second order of $-\log\beta_n(\varepsilon)$ is essentially given by $n\, (t^*_2(\eps) + C)$ and for $d\geq 3$ it is given by $n^{d-1} C$. 
\end{remark}	
\begin{remark}
	The stronger assumption of \Cref{condsoncoroll} holds in particular if the two states considered here are Gibbs states at different temperatures, under the condition that for some $\beta>0$,
	\begin{equation*} 
		\sum_{{\bf k} \in \ZZ^d} |{\bf k}|^d \big|\langle e_{\bf k}, \frac{\e^{-\beta H}}{1+\e^{-\beta H}} \,e_{\bf 0}\rangle \big| < + \infty,
	\end{equation*} This proves that for such systems we are in the case where \Cref{coro_ideal} holds.
\end{remark} 

\begin{proof}
	\Cref{eq_CARD} and \Cref{eq_psi} (which is point 2 of \Cref{cond-Bryc}) were proved in \cite{MHOF}. We now prove that the first and third parts of \Cref{cond-Bryc} hold as well. For any $n$, let $A_n=\frac{Q_n}{ \mathbb{I}_n - Q_n}$ and $B_n=\frac{R_n}{\mathbb{I}_n- R_n}$. We have $M^{-1}\leq A_n \leq 
	M$, and $M^{-1}\leq B_n \leq M$ for $M=(1-\delta)\delta^{-1}$.
	The operators~$A_n$ and $B_n$ are self-adjoint and we can therefore define their complex powers. Define $W_n^{(1-z )}= A_n^{\frac{1-z}2} B_n^{z} A_n^{\frac{1-z}2}$. As $A_n^z=A_n^s A_n^{i t}$, we have for $r\in(0,1)$ and any $z=s+i t$ in $B_\CC(0,r)$,
	\begin{align}
		\max(\|A_n^z\|,\|B_n^z\|)\leq M^{|s|} \quad \mbox{so that} \quad  \|W_n^{(1-z)}\|\leq M^{1+2r}. \label{ABmax}
	\end{align}
	As any differentiable function is locally Lipschitz, with Lipschitz constant given by the supremum norm of its derivative over each bounded interval, we also obtain the following bound
	\begin{align}
		\|A_n^s-A_n^{s+it}\|=\|A_n^{s+i0}-A_n^{s+it}\|\le \sup_{0\le u\le t}\|(\log A_n) A_n^{s+iu}\||t|\le r M^{|s|}\log M.
	\end{align}
	Proceeding in the same manner with $B_n$, we conclude that there exists a constant $C:=\log M$ such that
	\begin{align}
		\max\big(\|A_n^{s}- A_n^{s+i t}\|,\|B_n^{s}- B_n^{s+i t}\|\big)\leq C M^{|s|} r \quad \mbox{if } |t|<r.\label{diffABmax}
	\end{align}
	Using \Cref{ABmax}, \Cref{diffABmax} and the triangle inequality, we obtain
	\begin{align} 
		\big\|W_n^{({1-z})}-W_n^{(1-s)} \big\| &= \|A_n^{\frac{1-z}{2}}B_n^zA_n^{\frac{1-z}{2}}-A_n^{\frac{1-s}{2}}B_n^sA_n^{\frac{1-s}{2}}\|\label{eq_normcontrol}\\
		&\le \|A_n^{1/2}\|^2 \| A_n^{-z/2}(B_n^z- B_n^s)A_n^{-z/2}+ A_n^{-z/2} B_n^{s} A_n^{-z/2}-A_n^{-s/2}B_n^{s} A_n^{-s/2}\|\nonumber \\
		&\le M^{1+2|s|}  Cr+M \| (A_n^{-z/2}-A_n^{-s/2})B_n^sA_n^{-z/2}+A_n^{-s/2}B_n^{s}(A_n^{-z/2}-A_n^{-s/2})\|\nonumber \\
		&\le 3M^{1+2|s|}Cr\le 3M^{1+2r}Cr,\nonumber
	\end{align}
	Theorem VI.3.3 from \cite{Bhatia} implies that
	\begin{equation} \label{eq_spectralcontrol}
		\max_{\lambda\in \mathrm{sp}\,W_n^{({z})}}\min_{\mu\in \mathrm{sp}\,W_n^{({s})}}|\lambda-\mu| \leq   3 C M^{1+2r} r.
	\end{equation}
	Since $\mathrm{sp}\,W_n^{({1-s})}\subset [M^{-(1+2r)},M^{1+2r}]$, \Cref{eq_spectralcontrol} implies that, for $r$ small enough, the spectrum of $W_n^{(1-z)}$ remains in the half plane $\{z'\in \CC\, |\, \mathrm{Re}\,z' >0\}$ for any $n$ and $z\in B_\CC(0,r)$. Then considering the principal branch of the complex logarithm (which we simply denote by~$\log$), we can define $\tr \log (\mathbb{I}_n+W_n^{(1-z)})$ by Dunford-Taylor functional calculus (see \cite{Kato}) for any $z\in B_\CC(0,r)$, and it will be an analytic function of $z$. Then expression \eqref{psisfermions} gives us an analytic extension of $\frac1{n^d}\Psi_{1-z}(\rho_n|\sigma_n)$ to $z\in B_\CC(0,r)$, and a uniform bound is provided by 
	\begin{gather*}
		\max\big(\big|\tr\log(\II_n-Q_n)\big|,\big|\tr\log(\II_n-R_n)\big|\big)\leq n^d\,\log(1+M),\\
		|\tr\log(\II_n+W_n^{(1-z)})\big|\leq n^d\,\log(1+M^{1+2r}).
	\end{gather*}
	Therefore, points 1 and 3 of \Cref{cond-Bryc} hold.
	
	We now turn to the estimate of $D(\rho_n||\sigma_n)- n^d\,d(\hat\rho,\hat\sigma)$. 
	Using the fact that $Q_n$ and $R_n$ are truncated Toeplitz operators, as well as $\langle \e_\bfi,Q_n e_{\bfj}\rangle=\langle \e_\bfi,Q e_{\bfj}\rangle$ for any $\bfi,\bfj\in \{0,...,n-1\}^d$, we have
	\begin{align}
		\tr (Q_n\log R_n) &= \sum_{\bfi \in \{0,\ldots,n-1\}^d} \langle e_\bfi, Q_n \log R_n \, e_\bfi\rangle=  \sum_{\bfi,\bfj\in \{0,\ldots,n-1\}^d} \langle e_\bfi, Q e_\bfj\rangle \,\overline{\langle e_\bfi,\log R_n \, e_\bfj\rangle}\nonumber\\
		&=  \sum_{\bfi,\bfj\in \{0,\ldots,n-1\}^d} \langle e_{\bfi-\bfj}, Q e_\mathbf{0}\rangle \,\overline{\langle e_{\bfi-\bfj},\log R_n \, e_\mathbf{0}\rangle}\nonumber\\
		&=  \sum_{\bfk\in \{-(n-1),\ldots,+(n-1)\}^d}\big(\prod_{i=1}^d 2 (n-|k_i|)\big) \langle e_{\bfk}, Q e_\mathbf{0}\rangle \,\overline{\langle e_{\bfk},\log R_n \, e_\mathbf{0}\rangle}\nonumber\\
		&= (2n)^d  \sum_{\bfk\in \{-(n-1),\ldots,+(n-1)\}^d} \langle e_{\bfk}, Q e_\mathbf{0}\rangle \,\overline{\langle e_{\bfk},\log R_n \, e_\mathbf{0}\rangle}\nonumber\\
		&\quad + 2^d \sum_{\bfk\in \{-(n-1),...,(n-1)\}^d} \langle e_\bfk,Qe_\mathbf{0}\rangle \overline{\langle e_\bfk,\log R_n e_\mathbf{0}\rangle}\nonumber\\
&\qquad \times \sum_{l=0}^{d-1}n^{l}(-1)^{d-l}\sum_{1\le i_1<...<i_{d-l}\le d}\prod_{j=1}^{d-l}|k_{i_j}|\nonumber\\
		&=:n^d A_n+B_n\label{AnBn}
	\end{align}
	where $\bfk=(k_1,...,k_d)$. However, the condition \reff{condsoncoroll}, as well as the bound $\big|\langle e_{\bfk},\log R_n \, e_0\rangle\big|\leq \log (1-\delta)$ imply that 
	\begin{align*}
		&n\,\Big|\sum_{\bfk\notin \{-(n-1),\ldots,+(n-1)\}^d} \langle e_{\bfk}, Q e_\mathbf{0}\rangle \,\overline{\langle e_{\bfk},\log R_n \, e_\mathbf{0}\rangle}\Big|\le \log(1-\delta)\sum_{\bfk\notin \{-(n-1),\ldots,+(n-1)\}^d}|\bfk|\left| \langle e_{\bfk}, Qe_\mathbf{0}\rangle \right|\rightarrow0,
	\end{align*}
	so that 
	\begin{align*}
		\sum_{\bfk\in\ZZ^d} \langle e_{\bfk}, Q e_\mathbf{0}\rangle \,\overline{\langle e_{\bfk},\log R_n \, e_\mathbf{0}\rangle} -  \sum_{\bfk\in \{-(n-1),\ldots,+(n-1)\}^d} \langle e_{\bfk}, Q e_\mathbf{0}\rangle \,\overline{\langle e_{\bfk},\log R_n \, e_\mathbf{0}\rangle}=o(n^{-1})
	\end{align*}	
	Thus, the term $A_n$ in \Cref{AnBn} has the following asymptotic behavior:
	\begin{align*}
		A_n=2^d  \sum_{\bfk\in \ZZ^d} \langle e_{\bfk}, Q e_\mathbf{0}\rangle \,\overline{\langle e_{\bfk},\log R_n \, e_\mathbf{0}\rangle}+o(n^{-1}).
	\end{align*}	
	Similarly
	\begin{align*}
		|B_n|&\le 2^d \log(1-\delta) \sum_{l=0}^{d-1}n^l\sum_{\bfk\in\ZZ^d}|\langle e_\bfk,Qe_\mathbf{0}\rangle|\sum_{1\le i_1<...<i_{d-l}\le d}\quad\prod_{j=1}^{d-l}|k_{i_j}|\\
		&\le 2^d \log(1-\delta)\sum_{l=0}^{d-1}n^l\sum_{\bfk\in\ZZ^d}|\bfk|^d|\langle e_\bfk,Qe_\mathbf{0}\rangle|,
	\end{align*}
	which implies that $B_n=\mathcal{O}(n^{d-1})$. Hence, $\tr Q_n\log R_n = n^{d}A'+\mathcal{O}(n^{d-1})$ for some constant $A'$. Similarly we prove 
	that the three other terms of \Cref{eq_EntRevCAR} have the same asymptotic expansion, so that $D(\rho_n||\sigma_n)=n^d A^{\prime\prime}+\mathcal{O}(n^{d-1})$, for some constant $A^{\prime\prime}$. But from \Cref{eq_CARD} (which is a consequence of the multidimensional multivariate Szeg\"o limit theorem (see Lemma 3.1 of \cite{MHOF}), $A''=d(\hat{\rho},\hat{\sigma})$, and \Cref{strong} follows.
	\qed
\end{proof}

\paragraph{Acknowledgements}
N.D.~would like to thank Fernando Brandao, Mark Fannes and Will Matthews for helpful discussions. Y.P.~acknowledges the support of ANR contract ANR-14-CE25-0003-0, and would like to thank the Statistical Laboratory of the Centre for Mathematical Sciences at University of Cambridge (where part of this work was done) for hospitality. N.D.~and C.R.~are grateful to Felix Leditzky for help with the figures. All three authors are grateful to Vojkan Jaksic for helpful discussions. They would also like to thank Milan Mosonyi for his valuable comments on an earlier version of the paper.

\bibliography{biblio}

\appendix
\section{Vitali's theorem and Bryc's theorem} \label{app_Vitali}

In this section we state and prove Vitali's theorem, and prove Bryc's theorem (both in the case of one variable).

\begin{theorem}
Let $r>0$, and let $(F_n)_{n\in\NN}$ be a family of analytic functions on $B_\CC(0,r)$, such that
\begin{enumerate}
\item for any $x$ (real) in $(-r,+r)$, $F_n(x)\rightarrow_{n\rightarrow\infty}F(x)$,
\item one has the uniform bound $\sup_n \sup_{z\in B_\CC(0,r)} |F_n(z)|<\infty$,
\end{enumerate}
then 
\begin{enumerate}
\item the function $F$ so defined can be extended to an analytic function on $B_\CC(0,r)$, 
\item for any $k\in \{0\}\cup\NN$, the $k$-th derivative $F_n^{(k)}$ converges uniformly on any compact subset of $B_\CC(0,r)$ to the $k$-th derivative $F^{(k)}$.
\end{enumerate}
\end{theorem}

\noindent \begin{proof}
   For any arbitrary $r_0 \in (0,r)$, let $\Gamma_{r_0}$ denote the circle centered at the origin with radius $r_0$. We have for any $n$ and any $z$ in $B_\CC(0,r_0)$
\begin{equation} \label{eq_CauchyVitali}
F_n^{(k)}(z)=\frac{k!}{2\mathrm i \pi} \int_{\Gamma_{r_0}} \frac{F_n(u)}{(u-z)^{k+1}}\,\d u.
\end{equation}
Assumption $2$ implies that the family of derivatives $(F_n')_{n\in\NN}$ is uniformly bounded on the closed ball $\overline{B_\CC(0,r_0)}$. Since this holds for any arbitrary $r_0 \in (0,r)$, $(F_n)_{n\in\NN}$ is equicontinuous on any compact subset of $B_\CC(0,r)$.
The Arzela-Ascoli theorem implies that there exists a subsequence 
$\bigl( F_{\varphi(n)}\bigr)_{n}$ of $(F_n)_{n}$
which converges uniformly to an analytic limit $G$ on $\overline{B_\CC(0,r_0)}$. Necessarily, relation \eqref{eq_CauchyVitali} holds for $G$, which is therefore analytic on $B_\CC(0,r)$. By assumption $1$, on the interval $(-r, r)$, $G$ equals $F$ necessarily. This, uniquely determines $G$ 
on $B_\CC(0,r)$, so that $(F_n)_{n}$ itself converges to an analytic function which is an extension of $F$. The convergence of derivatives follows from \Cref{eq_CauchyVitali}.\qed \end{proof}

We now turn to a proof of Bryc's theorem, stated as \Cref{theo_Bryc}.

\begin{proof}
The analyticity of $H_n$ in a neighbourhood of zero shows that the random variables $X_n$ have moments of all orders. By a direct application of Vitali's theorem, the function $H$ is analytic on $B_\CC(0,r)$ and $\frac1{w_n}H_n$ converges uniformly on any compact subset of  $B_\CC(0,r)$ to $H$. Replacing $X_n$ with $X_n+ H_n'(0)$ we can assume that $H_n'(0)=0$ for all $n$. For $R>0$ denote by $\Gamma_R$ the circle centered at the origin, with radius $R$. Cauchy's formula gives, for any $0<r_0<r$, and $k\in\NN$
\begin{align*}
\frac{\d^k}{\d z^k} H(z)&= \frac{k!}{2\mathrm i \pi} \int_{\Gamma_{r_0/ \sqrt{w_n}}} \frac{H(u)}{u^{k+1}}\d u\\
&= \lim_{n\to\infty}\frac{k!}{w_n\,2\mathrm i \pi}\int_{\Gamma_{r_0/ \sqrt{w_n}}} \frac{H_n(u)}{u^{k+1}}\d u
\end{align*}
and a simple change of variable $u=u'/\sqrt{w_n}$ gives
\begin{align*}
H^{(k)}(0)&= \lim_{n\to\infty}w_n^{k/2-1}\frac{k!}{2\mathrm i \pi}\int_{\Gamma_{r_0}} \frac{H_n(u/\sqrt{w_n})}{u^{k+1}}\d u.
\end{align*}
However,
\begin{align*}
\frac{k!}{2\mathrm i \pi}\int_{\Gamma_{r_0}} \frac{H_n(u/\sqrt{w_n})}{u^{k+1}}\d u =\left. \frac{\d^k}{\d z^k}\right|_{z=0} \log \mathbb E \big(\exp - z \frac{X_n}{\sqrt{w_n}}\big), 
\end{align*}
so that
\begin{align}
H^{(k)}(0)&= \lim_{n\to\infty}w_n^{k/2-1}\left. \frac{\d^k}{\d z^k}\right|_{z=0} \log \mathbb E \big(\exp - z \frac{X_n}{\sqrt{w_n}}\big).
\end{align}
Therefore, for $k\geq 3$,
\[ \lim_{n\to\infty}\left. \frac{\d^k}{\d z^k} \right|_{z=0}\log \mathbb E \big(\exp - z \frac{X_n}{\sqrt{w_n}}\big)=0\]
whereas for $k=2$
\[\lim_{n\to\infty}\left.\frac{\d^2}{\d z^2} \right|_{z=0}\log \mathbb E \big(\exp - z \frac{X_n}{\sqrt{w_n}}\big) =\lim_{n\to\infty} \frac1{w_n} H_n''(0)=H''(0)\]
and by the choice that $H_n'(0)=0$ for any $n$,
\[\lim_{n\to\infty}\left. \frac{\d}{\d z}\right|_{z=0} \log \mathbb E \big(\exp - z \frac{X_n}{\sqrt{w_n}}\big) =\lim_{n\to\infty} \frac1{\sqrt {w_n}} H_n'(0)=0.\]
Therefore, the cumulant of order $k$ of $(X_n/\sqrt{w_n})_{n\in\NN}$ converges to the cumulant of order $k$ of the centered normal distribution $ \mathcal N(0,H''(0))$, so that the same is true of moments. Since the normal distribution is determined by its moments, the sequence $(X_n/\sqrt{w_n})_{n\in\NN}$ converges in law to $ \mathcal N(0,H''(0))$ (see Theorem 30.2 and Example 30.1 in \cite{Billingsley}).
\qed \end{proof}

\section{Proof of \Cref{boundsD} and \Cref{cor1}} \label{app_boundseNP}

\paragraph{Proof of \Cref{boundsD}} We first prove the upper bound in \Cref{ineq}. In order to do so, we use the following inequality (see \cite{Jaksicetal}): For any $A,B \in \cP_{\RR_+}(\cH)$ and $s$ in $[0,1]$,
\begin{align}\label{matrixineq}
\frac{1}{2}(\tr A+\tr B-\tr |A-B|)\le \tr A^{1-s}B^s.
\end{align}	

For a given test $T$,
\begin{align*}
\operatorname{\es}(A,B,T)&=\tr(A(\mathbb{I}- T))+\tr(B T)\\
&=\tr A+\tr(T(B-A)).
\end{align*}
Therefore, by \Cref{lemmaproj}, 
\begin{align}
\operatorname{\est}(A,B)&=\tr A-\tr(A-B)P_{+}(A-B)\nonumber\\
&=1/2(\tr A+\tr B-\tr|A-B|),\label{eqq}
\end{align}		
where $P_{\RR_+}(A-B)$ denotes the projection onto the support of $(A-B)_+$. 
From \Cref{phi-s-inner} we have that $\langle \Omega_B, \Delta^s_{A|B}\Omega_B\rangle =\tr(A^s B^{1-s})$ where $\Omega_B = B^{1/2}$. Hence using \Cref{matrixineq} we obtain
\begin{align*}
\langle \Omega_{B}, \Delta^s_{A|B} \Omega_{B}\rangle\ge 1/2(\tr A+\tr B-\tr|A-B|)=\operatorname{\est}(A,B).
\end{align*}
\medskip

Next we prove the lower bound in \Cref{ineq}. By \Cref{lemmaproj}
we know that the infimum in \Cref{e-symm} of $\est$ 
is attained for a test given by an orthogonal projection which we denote by $P$.
For such a test, $P$, using spectral decompositions of $A$ and $B$ we can write
\begin{align*}
\tr(A (\mathbb{I}-P))&=\sum_{\lambda\in \operatorname{sp}(A)}\lambda\tr((\mathbb{I}-P)P_\lambda(A)(\mathbb{I}-P))\\
&=\sum_{(\lambda,\mu)\in\operatorname{sp}(\rho)\times\operatorname{sp}(B)}\lambda\mu^{-1}\tr(B (\mathbb{I}-P) P_{\lambda}(A)(\mathbb{I}-P)P_\mu(B)).
\end{align*}	
Similarly,
\begin{align*}
\tr(B P)=\sum_{(\lambda,\mu)\in\operatorname{sp}(A)\times\operatorname{sp}(B)}\tr(B PP_{\lambda}(A)PP_\mu(B)).		
\end{align*}
where we have used the fact that $\sum_\lambda P_\lambda(A) = \mathbb{I} = \sum_\mu P_\mu(B)$ since $A, B \in \cP_{\RR_+}(\cH)$.		
Therefore as for any $\kappa\ge 0$, denoting $\bP = \mathbb{I} - P$, we have 
\begin{align*}
\kappa \bP P_\lambda(A)\bP+ PP_{\lambda}(A)P &=\frac{\kappa}{1+\kappa}P_\lambda(A)+\frac{1}{1+\kappa}(1-(1+\kappa)\bP)P_{\lambda}(A)(1-(1+\kappa)\bP)\\
&\ge \frac{\kappa}{1+\kappa}P_\lambda(A)
\end{align*}
and hence we obtain
\begin{align*}
\tr(B P)+\tr(A(\mathbb{I}-P))&\ge\sum_{(\lambda,\mu)\in \operatorname{sp}(A)\times\operatorname{sp}(B)}\frac{\lambda/\mu}{1+\lambda/\mu}\tr(B P_\lambda(A)P_\mu(B))\\
&=\langle \Omega_B,\Delta_{A|B}(1+\Delta_{A|B})^{-1}\Omega_B\rangle
\end{align*}
and the lower bound in \Cref{ineq} follows.
\qed

\paragraph{Proof of \Cref{cor1}}
\Cref{boundsD} implies that for any $\theta \in\mathbb{R}$,
\begin{align}
\operatorname{\est}(\sigma,\operatorname{e}^{-\theta} \rho)&\ge \langle \Omega_{\e^{-\theta}\rho},\Delta_{\sigma|\e^{-\theta}\rho}(1+\Delta_{\sigma|\e^{-\theta}\rho})^{-1}\Omega_{\e^{-\theta}\rho}\rangle\nonumber\\
&=\operatorname{e}^{-\theta}\langle\Omega_{\rho},\operatorname{e}^{\theta}\Delta_{\sigma|\rho}(1+\operatorname{e}^{\theta}\Delta_{\sigma|\rho})^{-1}\Omega_\rho\rangle\nonumber\\
&=\int_{\mathbb{R}}\frac{\operatorname{e}^{-\theta}}{1+\operatorname{e}^{x-\theta}}d\mu_{\sigma|\rho}(x) 
\end{align}
and the result follows from Markov's inequality.
\qed

\section{Proof of \Cref{prop_1}} \label{app_Prop1}
Let the spectral decompositions of the states $\rho$ and $\sigma$ be given as follows,	
\[\rho = \sum_{\lambda\in\spec(\rho)} \lambda P_\lambda(\rho) \qquad \sigma= \sum_{\mu\in\spec(\sigma)} \mu P_\mu(\sigma)\]
Denote for any $\lambda\in \spec(\rho)$:
\[Q_{\lambda}:= P_{(-\infty,\lambda/L]}(\sigma)=  \sum_{\mu\in\spec(\sigma):L\mu\le \lambda} P_{\mu}(\sigma). \]
$Q_\lambda$ then increases with $\lambda$. The operator 
\[ \tilde T :=  \sum_{\lambda\in \spec(\rho),\mu\in\spec(\sigma):L\mu\le\lambda}  P_\mu(\sigma) P_\lambda(\rho) P_\mu(\sigma)\]
is well-defined and nonnegative. We define our test $T$ by the projection onto the support of $\tilde T$. 
Now fix $\lambda\in\spec(\rho)$. For any $\varphi\in \mathrm{ran}(P_\lambda(\rho))$ of unit norm, one has $P_\lambda(\rho)\ge |\varphi\rangle\langle \varphi|$, where $|\varphi\rangle\langle \varphi|$ denotes the projector onto $\operatorname{span}(\varphi)$, 
and therefore for any $\mu\in\spec(\sigma)$ such that $L\mu\le \lambda$ one has $\tilde{T} \ge |P_\mu(\sigma)\varphi\rangle\langle P_\mu(\sigma)\varphi|$.
This implies that $P_\mu(\sigma)\varphi\in \operatorname{supp}(\tilde{T})$. Hence $$T\ge \frac{\ketbra{P_\mu(\sigma) \varphi}{P_\mu(\sigma) \varphi}}{\|P_\mu(\sigma) \varphi\|^2}. $$ 
Since $\{P_\mu(\sigma)\varphi\}_{\mu\in\spec(\sigma)}$ form a family of orthogonal vectors, the following inequality holds,
\[ T \geq \sum_{\mu\in\spec(\sigma):L\mu\le \lambda} \frac{\ketbra{P_\mu(\sigma) \varphi}{P_\mu(\sigma) \varphi}}{\|P_\mu(\sigma) \varphi\|^2}.\]
Moreover, since
\[\tr \big(\ketbra \varphi\varphi\, \ketbra{P_\mu(\sigma) \varphi}{P_\mu(\sigma) \varphi}\big)=  \|P_\mu(\sigma) \varphi\|^4 \]
we have that
\[\tr \big(\ketbra \varphi\varphi\, T\big) \geq \sum_{\mu\in\spec(\sigma):L\mu\le\lambda}  \|P_\mu(\sigma) \varphi\|^2 =  \sum_{\mu\in\spec(\sigma):L\mu\le\lambda}\tr \big(\ketbra \varphi\varphi\, P_\mu(\sigma)\big) \]
for $\varphi\in\operatorname{ran}(P_\lambda(\rho))$. Thus by writing $P_{\lambda}(\rho)=\sum_{i=1}^{\operatorname{dim}(\mathrm{ran}(P_\lambda(\rho)))}|\varphi_i\rangle\langle \varphi_i|$, where $\{\varphi_i\}_i$ is an orthonormal basis for $\mathrm{ran}(P_{\lambda}(\rho))$,
\[\tr \,P_\lambda(\rho) T \geq \sum_{\mu\in\spec(\sigma):L\mu\le\lambda}  \tr \,P_\lambda(\rho)\, P_\mu(\sigma) \]
and
\begin{align*}
\tr \rho T \geq \sum_{\lambda\in\spec(\rho),\mu\in\spec(\sigma):L\mu\le\lambda}  \lambda  \tr \,P_\lambda(\rho)\, P_\mu(\sigma)& = \sum_{\lambda\in\spec(\rho),\mu\in\spec(\sigma):L\mu\le\lambda} \tr \big(\rho^{1/2}\,P_\lambda(\rho)\, \rho^{1/2}\,P_\mu(\sigma)\big)\\
&=   \big\langle{\Omega_\rho},{P_{[L,+\infty)}(\Delta_{\rho|\sigma})\,\Omega_\rho}\big\rangle.
\end{align*}
Since $\Delta_{\rho|\sigma}$ is a positive operator, this proves the first inequality in \eqref{eq_alphabeta}.
\smallskip

To prove the second inequality in \eqref{eq_alphabeta}, first observe that for any $\lambda\in\spec(\rho)$, $\mu\in\spec(\sigma)$:
\begin{align}
\mathrm{ran}(P_\mu(\sigma) P_\lambda(\rho))=\mathrm{supp}(P_\mu(\sigma) P_\lambda(\rho)P_\mu(\sigma)).
\end{align}
Indeed, using the decomposition $P_{\lambda}(\rho)=\sum_{i=1}^{\operatorname{dim}(\mathrm{ran}(P_\lambda(\rho)))}|\varphi_i\rangle\langle \varphi_i|$, where $\{\varphi_i\}$ is an orthonormal basis for $\mathrm{ran}(P_{\lambda}(\rho))$, we obtain
\begin{align*}
P_\mu(\sigma)P_\lambda(\rho)P_\mu(\sigma)=\sum_{i=1}^{\operatorname{dim}\mathrm{ran}(P_\lambda(\rho))}|P_\mu(\sigma)\varphi_i\rangle\langle P_\mu(\sigma)\varphi_i|,
\end{align*}
which implies that the support of $P_\mu(\sigma)P_\lambda(\rho)P_\mu(\sigma)$ is the span of the set of vectors $\{P_\mu(\sigma)\varphi_i\}_i$, and is hence the range of $P_\mu(\sigma)P_{\lambda}(\rho).$ 
Then denote for $\lambda\in \spec(\rho)$,
\begin{equation} \label{eq_deftildeAk}
\tilde{\mathcal{A}}_\lambda=\bigvee_{\mu\in\spec(\sigma):\lambda/\mu\ge L} \mathrm{ran}( P_\mu(\sigma) P_\lambda(\rho))=\bigvee_{\mu\in\spec(\sigma):\lambda/\mu\ge L}\mathrm{supp}(P_\mu(\sigma) P_\lambda(\rho)P_\mu(\sigma)).
\end{equation}
To complete the proof we employ a Gram-Schmidt orthogonalization procedure as in \cite{L14}. We first order the eigenvalues of $\rho$ and denote them as $\lambda_1, \lambda_2, \ldots, \lambda_{d_\rho}$
where $d_{\rho}$ is the number of different eigenvalues of $\rho$, and $\lambda_1 > \lambda_2 >\ldots > \lambda_{d_\rho}$. Define a family $(\mathcal{A}_i)_{1\le i\le d_\rho}$ of subspaces of $\mathcal H$ by $\mathcal{A}_1 := \tilde{\mathcal{A}}_{\lambda_1}$ and, $\forall\, 1 \leq k \leq d_\rho-1$, 
\begin{align}\label{AtildeA}
\mathcal{A}_{k+1}:= \big( \tilde{\mathcal{A}}_{\lambda_1}+\ldots + \tilde{\mathcal{A}}_{\lambda_{k+1}}\big)\cap \big(\tilde{\mathcal{ A}}_{\lambda_1}+\ldots + \tilde{\mathcal{ A}}_{\lambda_k}\big)^\perp.
\end{align}
The subspaces $\mathcal{A}_k$ are mutually orthogonal by construction.
Moreover, it can be shown by induction that for any $k$,
\begin{equation} \label{eq_sums}
\mathcal{A}_1+\ldots +\mathcal{A}_k = \tilde{\mathcal{ A}}_{\lambda_1} + \ldots + \tilde{\mathcal{ A}}_{\lambda_k}.
\end{equation}
By definition, $T$ is the projection on
\begin{align}
\bigvee_{\lambda\in \spec(\rho),\mu\in\spec(\sigma):\lambda\ge  L\mu}\mathrm{supp}(P_\mu(\sigma) P_\lambda(\rho)P_\mu(\sigma)) &=\bigvee_{\lambda\in\spec(\rho),\mu\in\spec(\sigma):\lambda\ge L\mu} \mathrm{ran}(P_{\mu}(\sigma) P_\lambda(\rho)) 
\nonumber\\
&= \bigvee_{\lambda} \tilde{\mathcal{A}}_{\lambda} = \bigoplus_{k=1}^{d_\rho} \mathcal{A}_k,
\end{align}
so that $T = \sum_k^{d_\rho} P_{\mathcal{A}_k}$ where, for each $k$, $P_{\mathcal{A}_k}$ is the orthogonal projector onto $\mathcal{A}_k$. Note that
\begin{align}\label{AAtildineq}
\mathcal{A}_{k}\subset{\tilde{\mathcal{A}}}_{\lambda_k}\Rightarrow\operatorname{dim}\mathcal{A}_k\le \operatorname{dim}\tilde{\mathcal{A}}_{\lambda_k}.
\end{align}
Using the definition of  $\tilde{\mathcal{A}}_{\lambda_k}$, one finally gets
\begin{align}\label{okey}
 \tr P_{\mathcal{A}_k}=\operatorname{dim}\mathcal{A}_k\le \operatorname{dim}\tilde{\mathcal{A}}_{\lambda_k}\le \tr P_{\lambda_k}(\rho).
 \end{align}
We then obtain
\begin{align}
\tr\, \sigma T &= \sum_{\mu\in\spec(\sigma)}\sum_{k=1}^{d_\rho} \mu \,\tr \left(P_{\mu}(\sigma) P_{\mathcal{A}_k}\right) \nonumber\\
&= \sum_{k=1}^{d_\rho}\sum_{\mu\in\spec(\sigma):\atop{\lambda_k/\mu\ge L}} \mu \,\tr  \left(P_{\mu}(\sigma) P_{\mathcal{A}_k}\right)\nonumber\\
& \leq L^{-1} \sum_{\mu\in\spec(\sigma):\atop{\lambda_k/\mu\ge L}} \sum_{k=1}^{d_\rho} \lambda_k \tr \left(P_{\mu}(\sigma) P_{\mathcal{A}_k} \right) \nonumber\\
&\leq L^{-1} \sum_{\mu\in\spec(\sigma)} \sum_{k=1}^{d_\rho} \lambda_k \tr \left(P_{\mu}(\sigma) P_{\mathcal{A}_k}\right)= L^{-1} \sum_{k=1}^{d_\rho} \lambda_k \tr \left(  P_{\mathcal{A}_k}\right)\nonumber\\
& \leq L^{-1} \sum_{k=1}^{d_\rho} \lambda_k \tr \left(  P_{\lambda_k}(\rho)\right)=L^{-1}\tr(\rho)=L^{-1},\nonumber
\end{align}
where the second line follows from the fact that $P_{\mu}(\sigma) P_{\mathcal{A}_k} = 0$ unless $\lambda_k/\mu\ge L$ (by (\ref{eq_deftildeAk}) and (\ref{AAtildineq})), and the last one from \eqref{okey}. This concludes the proof. \qed

\end{document}